\documentclass[review]{elsarticle}

\usepackage{hyperref}

\journal{Journal of \LaTeX\ Templates}

\usepackage{graphicx}
\usepackage{subfigure}
\usepackage{mathtools}

\usepackage{tcolorbox}

\usepackage{xcolor}

\usepackage{amsthm}

\newtheorem{theorem}{Theorem}
\newtheorem{lemma}{Lemma}
\newtheorem{definition}{Definition}

\usepackage{amsmath}
\makeatletter
\providecommand*{\cupdot}{%
	\mathbin{%
		\mathpalette\@cupdot{}%
	}%
}
\newcommand*{\@cupdot}[2]{%
	\ooalign{%
		$\m@th#1\cup$\cr
		\hidewidth$\m@th#1\cdot$\hidewidth
	}%
}
\makeatother

\usepackage[nameinlink, capitalise]{cleveref}
\crefname{theorem1}{Theorem}{}
\crefname{corollary}{Corollary}{}
\crefname{lemma1}{Lemma}{}
\crefname{claim}{Claim}{}
\crefname{definition}{Definition}{}
\crefname{observation1}{Observation}{}
\crefname{section}{Section}{}
\crefname{figure}{Figure}{}
\crefname{problems}{Problem}{}

\usepackage{algorithm,float}
\usepackage[noend]{algpseudocode}
\usepackage{lipsum}
\usepackage{caption}

\makeatletter
\newenvironment{breakablealgorithm}
{
	\begin{center}
		\refstepcounter{algorithm}
		\hrule height.8pt depth0pt \kern2pt
		\renewcommand{\caption}[2][\relax]{
			{\raggedright\textbf{\ALG@name~\thealgorithm} ##2\par}%
			\ifx\relax##1\relax 
			\addcontentsline{loa}{algorithm}{\protect\numberline{\thealgorithm}##2}%
			\else 
			\addcontentsline{loa}{algorithm}{\protect\numberline{\thealgorithm}##1}%
			\fi
			\kern2pt\hrule\kern2pt
		}
	}{
		\kern2pt\hrule\relax
	\end{center}
}
\makeatother

\usepackage{diagbox}
\usepackage{amsmath,amsfonts,amssymb,amsthm}
\usepackage{bm}
\newcommand*{\B}[1]{\ifmmode\bm{#1}\else\textbf{#1}\fi}

\newcommand{\pg}{$\mathsf{Manhattan~network}$}
\newcommand{\planar}{$\mathsf{planar~Manhattan~network}$}

\newcommand{\pgt}{\textit{Manhattan network}}

\newcommand{\mmn}{\textsf{MMN}}

\usepackage{bbding}
\usepackage[framemethod=tikz]{mdframed}
\newmdenv[backgroundcolor=orange!10,
topline=false,
bottomline=false,
rightline=false,
skipabove=\topsep,
skipbelow=\topsep
]{siderules}

\usepackage{tcolorbox}
\usepackage{xcolor}
\definecolor{Sapphire}{RGB}{15,82,186}

\DeclarePairedDelimiterX{\norm}[1]{\lVert}{\rVert}{#1}

\begin{document}

\begin{frontmatter}


\title{Linear Size Planar Manhattan Network for Convex Point Sets}





%
%
%
%



\author[mymainaddress]{Satyabrata Jana\corref{mycorrespondingauthor}}
\cortext[mycorrespondingauthor]{Corresponding author}
\ead{satyamtma@gmail.com}

\author[mysecondaryaddress]{Anil Maheshwari}
\ead{anil@scs.carleton.ca}

\author[mymainaddress]{Sasanka Roy}
\ead{sasanka.ro@gmail.com}

\address[mymainaddress]{Indian Statistical Institute, Kolkata, India}
\address[mysecondaryaddress]{School of Computer Science, Carleton University, Ottawa, Canada}

\begin{abstract}
	

 Let $G = (V, E)$ be an edge weighted \emph{geometric graph} such that every edge is horizontal or vertical.
 The weight of	an edge  $uv \in E$ is its length. Let $ W_G (u,v)$ denote the length of a shortest path between a pair of vertices $u$ and $v$ in $G$.
The graph $G$ is said to be a \pg~for a given point set $ P $ in the plane
if $P \subseteq V$ and $\forall p,q \in P$, $ W_G (p,q)=\norm{pq }_1$. In addition to $ P$,  graph $G$ may also include a set $T$ of \emph{Steiner points} in its vertex set $V$. In the \pg~problem, the objective is to	construct a \pg~of small size for a set of $ n $ points. This problem was first considered by Gudmundsson et al.\cite{gudmundsson2007small}. They give a construction of a \pg~of size $\Theta(n \log n)$ for general point set in the plane. We say a \pg~is planar if it can be embedded in the plane without any edge crossings. In this paper, we construct a linear size \planar~for convex point set in linear time using $\mathcal{ O}(n)$ Steiner points. We also show that,	even for convex point set, the construction in Gudmundsson et al. \cite{gudmundsson2007small} needs $\Omega (n \log n)$ Steiner points and the network may not be planar.
\end{abstract}


\end{frontmatter}


\section{Introduction}
\noindent In computational geometry, constructing a minimum length \pg ~is a well-studied area \cite{gudmundsson1999approximating}. A graph $G=(V,E)$ is said to be a \pg~ for a given point set $ P $ in the plane if $P \subseteq V$ and $\forall p,q \in P$, $ W_G (p,q)=\norm{pq }_1$, where $ W_G (u,v)$ denotes the length of a shortest path between a pair of vertices $u$ and $v$ in $G$. The graph $G$ may also include a set $T$ of Steiner points in its vertex set $V$. A Minimum \pg~(\mmn) problem on $P$ is to construct a  \pg~of minimum possible length. Below in \cref{fig:m1} and \cref{fig:m2}, we show examples of a \pg~and a Minimum \pg~on the same set of points.

\begin{figure}[ht!]
	\centering
	\subfigure[]
	{
		\includegraphics[scale=0.4]{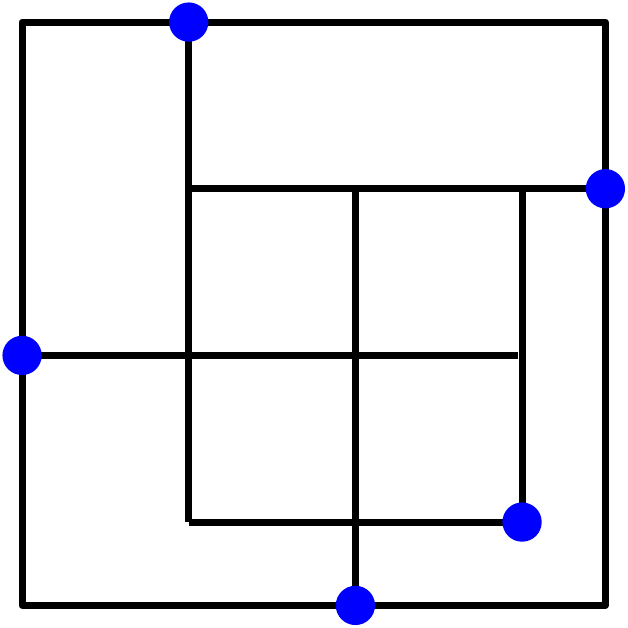}
		\label{fig:m1}
	}
	\hspace{18mm}
	\subfigure[]
	{
		\includegraphics[scale=0.4]{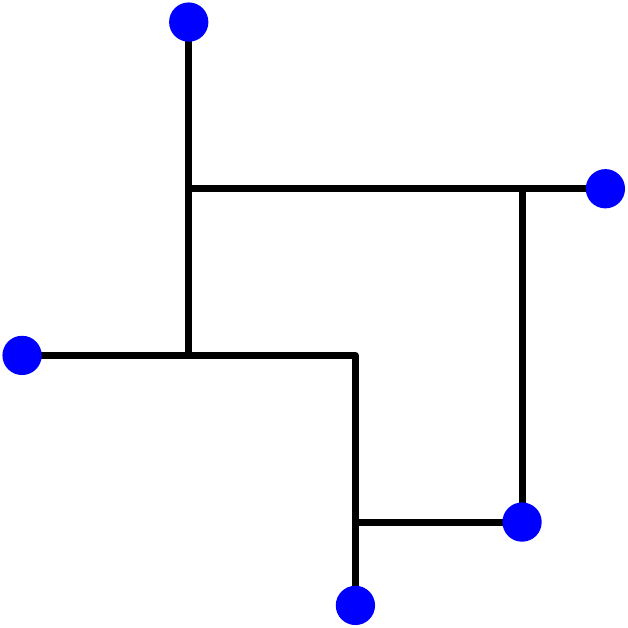}
		\label{fig:m2}
	}
	
	\caption{(a) A Manhattan network, and (b) A minimum Manhattan network.} 
	\label{fig-m}
\end{figure}

\mmn~problem has a wide number of applications in city planning, network layouts, distributed algorithms \cite{narasimhan2007geometric}, VLSI circuit design \cite{gudmundsson1999approximating}, and computational biology \cite{lam2003picking}. The MMN problem was first introduced in 1999 by Gudmundsson et al. \cite{gudmundsson1999approximating}.    Several approximation  algorithms (with factors 4 \cite{gudmundsson2007small}, 2   \cite{kato2002improved}, and 1.5 \cite{seibert20051}) with time complexity $\mathcal{O}(n ^3)$ have been proposed in the last few years. Also, there are $\mathcal{O}(n \log n)$ time  approximation  algorithms with factors 8 \cite{gudmundsson2007small}, 3 \cite{benkert2004minimum}, and 2 \cite{guo2008greedy}. Recently Chin et al. \cite{chin2011minimum} proved that the decision version of the \mmn~problem is strongly NP-complete. Recently,  Knauer et al. \cite{knauer2011fixed} showed that this problem is fixed parameter tractable.

In 2007, Gudmundsson et al. \cite{gudmundsson2007small} considered a variant of the \mmn~problem where the goal is to minimize  the number of vertices(Steiner) and edges.  In $\mathcal{O}(n \log n)$ time, they construct a \pg~with $\mathcal{O}(n \log n)$ vertices and edges using divide and conquer strategy. They also proved that there are point sets in $\mathbb{R}^2$ where every \pg~on these points will need $\Omega(n \log n)$ vertices and edges.



 A set of points is said to be a convex point set  if all of the points are vertices of their convex hull. A {\em plane} \pg~is a \pg~without non-crossing edges. Gudmundsson et al. \cite{gudmundsson2007small} showed that there exists a convex point set for which a {\em plane} \pg~requires $\Omega(n^2)$ vertices and edges. Now we explain the construction of the plane \pg~given by  Gudmundsson et al. \cite{gudmundsson2007small}. To keep it simple, we would use the same notations as they use. Let $P$ be a set  of points in the plane as follows:

$$P=\bigcup_{i=1}^{n-1}\{(\frac{i}{n},0),(\frac{i}{n},1),(0,\frac{i}{n}),(1,\frac{i}{n})\}$$


If $G$ is a plane \pg~of $P$ then there must be a shortest $L_1$ path between every pair of points $(\frac{i}{n},0),(\frac{i}{n},1)$ and $(0,\frac{i}{n}), (1,\frac{i}{n})$. These paths need to be orthogonal straight line segments because in the first case the $x$-coordinates are the same and in the second case the $y$-coordinates are the same. This would force us to add Steiner points at all the $\Theta(n^2)$ intersection points. For an illustration, see \cref{fig-pp1}.

\begin{figure}[ht!]
	
	\subfigure[]
	{
		\includegraphics[scale=0.31]{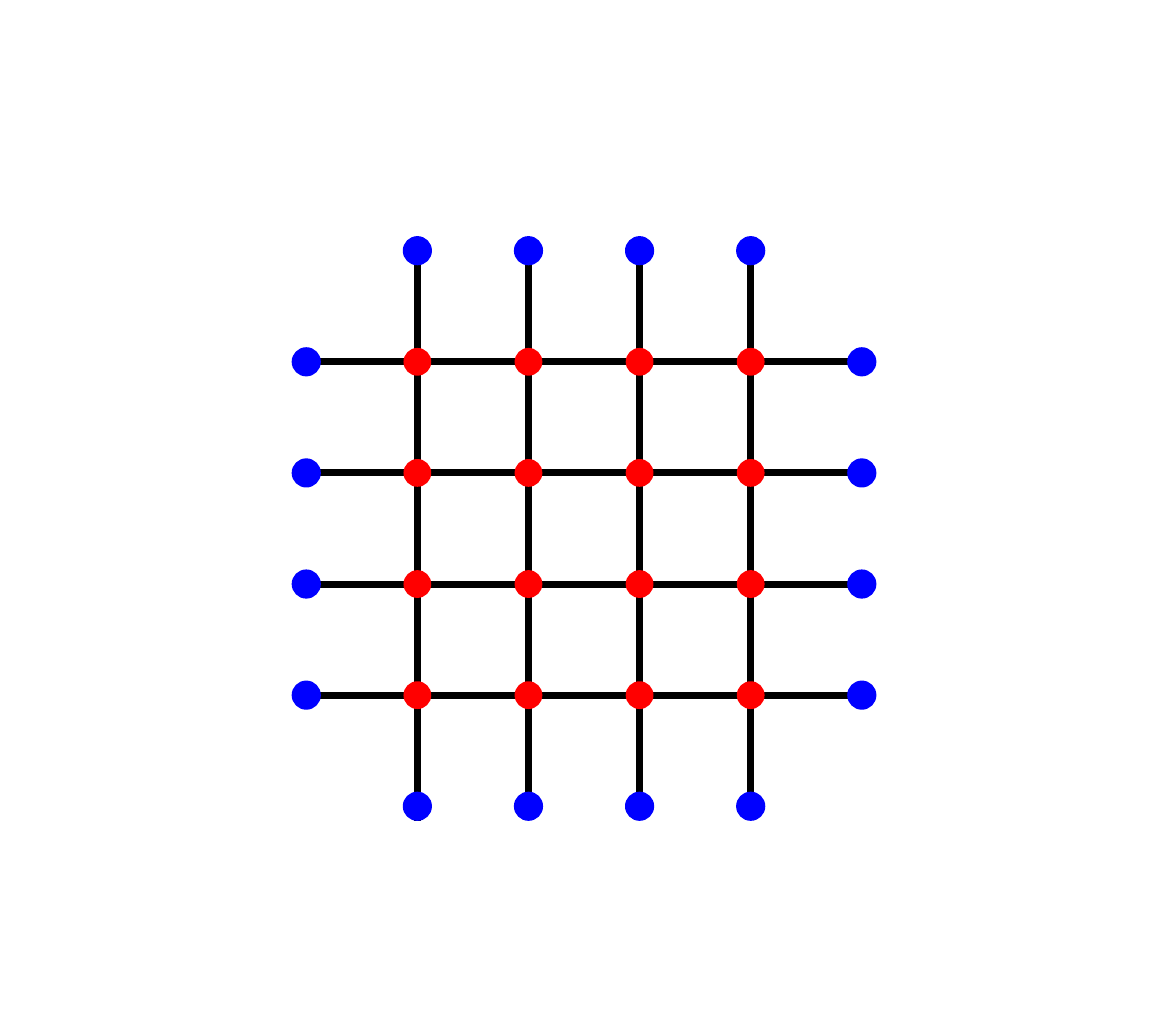}
		\label{fig-pp1}
	}	
	\subfigure[]
	{
		\includegraphics[scale=0.31]{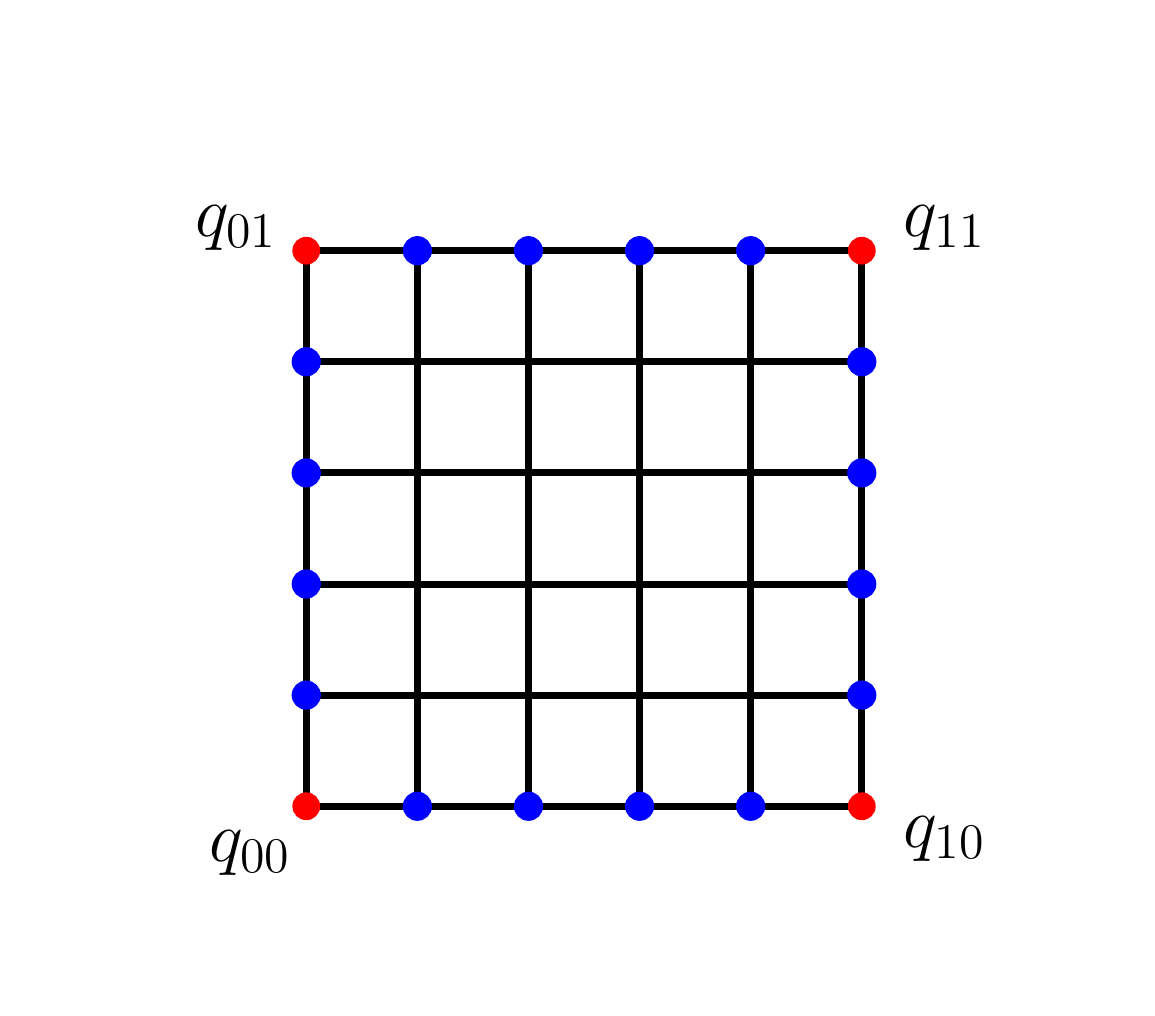}
		\label{fig-pp2}
	}
	\subfigure[]
	{
		\includegraphics[scale=0.31]{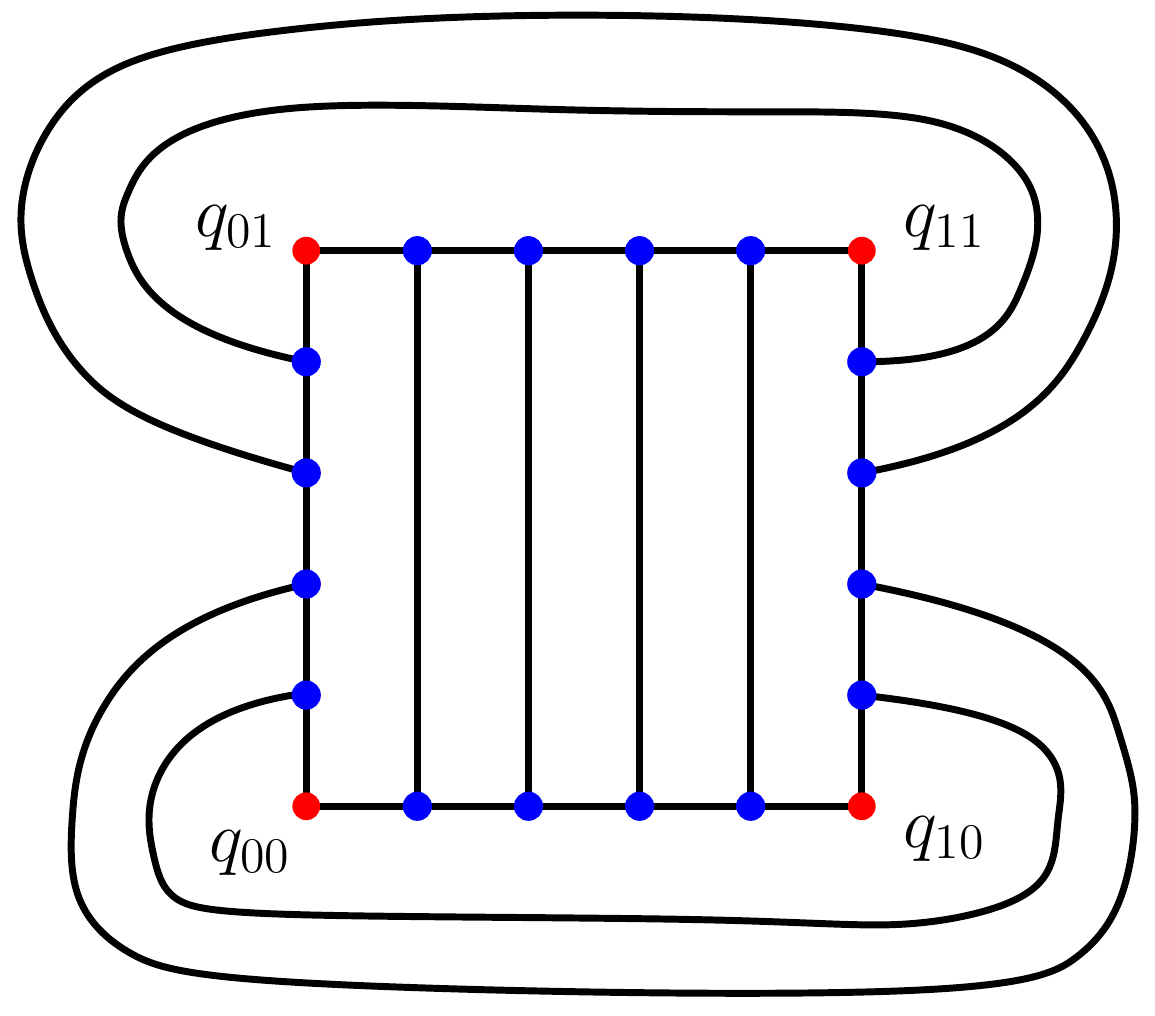}
		\label{fig-pp3}
	}

	\caption{(a) Lower bound construction of plane \pg~of $P$ (b) Planar  \pg~$G^*$ of $P$ and (c) Planar embedding of $G^*$. Blue circles represent the points in $P$ and red circles represent
		Steiner points.} 
	\label{fig-pp}
\end{figure}

A natural question that arises is what if we want the network to be planar (and not necessarily plane).  We say a Manhattan network is planar if it can be embedded in the plane without any edge crossings. For the above example, we can construct a \planar~$ G=(V=P\cup T, E) $ of $ \mathcal{O}(n) $ size as follows: Note that, $P$ lies on the boundary of a square $Q=[(0,0),(0,1)]\times[(1,0),(1,1)]$ (see \cref{fig-pp2}). We add four Steiner points $q_{00}=(0,0)$, $q_{01}=(0,1)$,  $q_{10}=(1,0)$, $q_{11}=(1,1)$, and we  define $T=\{q_{00},q_{01},q_{10},q_{11}\}$. For $ i= 1, 2, \dots , n-1 $, we add the edges between every pair of consecutive points (including these four Steiner points) on the boundary of $Q$. We also add the edges between every pair of points $(\frac{i}{n},0),(\frac{i}{n},1)$ and $(0,\frac{i}{n}), (1,\frac{i}{n})$.
To show that $G$ is a Manhattan network, we prove that $\forall p,q \in P$, $ W_G (p,q)=\norm{pq }_1$. Following is the description of all these paths in $G$. The paths between every pair of points $(\frac{i}{n},0),(\frac{i}{n},1)$ and $(0,\frac{i}{n}), (1,\frac{i}{n})$ is a straight line segment (horizontal and vertical). The paths between every pair of points $(\frac{i}{n},0),(0,\frac{j}{n})$ go through $q_{00}$. Likewise, the paths between every pair of points $(0,\frac{i}{n}),(\frac{j}{n},1)$ go through $q_{01}$, the paths between every pair of points $(0,\frac{i}{n}),(\frac{j}{n},1) $ go through $q_{10}$, the paths between every pair of points $(\frac{i}{n},1),(1,\frac{j}{n}) $ go through $q_{11}$. Between every pair of points $(\frac{i}{n},0),(\frac{j}{n},1) $ there exists a path through $(\frac{i}{n},1)$. Similarly, between every pair of points $(0,\frac{i}{n}),(1,\frac{j}{n}) $ there exists a path through $(1, \frac{i}{n})$.
To show that $G$ is planar, we provide its planar embedding. For the planar embedding of $G$, we keep the edges between every pair of points $(\frac{i}{n},0)$ and $(\frac{i}{n},1)$ inside the interior face of $Q$ and draw the edges between $(0,\frac{i}{n})$ and $(1,\frac{i}{n})$ in the exterior face of $Q$. For an illustration, see \cref{fig-pp3}.

A closely related problem is to construct geometric spanner from a given point set. For a real number $t\geq 1$, a geometric graph $G=(S,E)$ is a $t$-spanner of $S$ if for any two points $p$ and $q$ in $S$, $W_{G}(p,q) \leq t|pq|$. The {\em{stretch factor} } of $G$ is the smallest real number $t$ such that $G$ is a $t$-spanner of $S$. A large number of algorithms have been proposed for constructing $t$-spanners for any given point set \cite{narasimhan2007geometric}. Keil et al. \cite{Keil1989} showed that the Delaunay triangulation of $S$ is a 2.42-spanner of $ S $. For convex point sets, Cui et al. \cite{Cui2011}  proved that the Delaunay triangulation has a stretch factor of at most 2.33. Xia \cite{Xia2013}  provides a 1.998-spanner for general point sets. Steiner points have also been used for constructing spanners. For example, Arikati et al. \cite{arikati1996planar} use Steiner points to answer exact shortest path queries between any two vertices of a geometric graph.  Authors \cite{arikati1996planar} consider the problem of finding an obstacle-avoiding $L_{1}$ path between a pair of query points in the plane. They find a $ (1+\epsilon) $ spanner with space complexity $ \mathcal{O}(n^2/\sqrt{r}) $, preprocessing time $\mathcal{O}(n^2/\sqrt{r}) $ and $\mathcal{O}(\log n + \sqrt{r} )$ query time, where $ \epsilon $ is an arbitrarily small positive constant and $r$ is an arbitrary integer, such that $1 < r < n$. 
Recently, Amani et al. \cite{amani2016plane} show how to compute a plane 1.88-spanner in $L_2$ norm for convex point sets in $\mathcal{O}(n)$ time without using Steiner points.   For a general point set of size $ n $, Gudmundsson et al. \cite{gudmundsson2007small} construct a $\sqrt2$-spanner (may not be planar) in $L_2$ norm and its size is $\mathcal{O}(n \log n)$. But as a corollary of our construction in this paper, for a convex point set, we obtain a planar $\sqrt2 $ spanner in $L_2$ norm using $\mathcal{O}(n)$ Steiner points. The \mmn~problem for a point set is same as the problem of finding a 1-spanner in $L_1$-metric \cite{chin2011minimum}. Given a rectilinear polygon with $n$ vertices, in linear time, Schuierer \cite{schuierer1996optimal} constructs a data structure  that can report the shortest path (in  $ L_1$ -metric) for any pair of query points in that polygon  in $\mathcal{O}(\log+k)$ time where $k$ is the number of segments in the shortest path. De Berg \cite{de1991rectilinear} shows that given two arbitrary points inside
a polygon, the $L_1$-distance between them can be reported in $\mathcal{ O}(\log n)$ time. In this paper, we consider the following problem.

\begin{tcolorbox}[colback=red!5!white]
	\noindent {\bf{\color{red!50!black} \pgt~problem}}\\
	{\bf Input:} A set $S$ of $n$ points in convex position.\\
	{\bf Goal:} To construct a linear size \planar.
\end{tcolorbox}


\newpage\subsection{Our Contributions}
\noindent

$\bullet$ In linear time, we construct a \planar~$G$  for a convex point set $ S $ of size $ n $.  $G $ uses  $ \mathcal{O}(n) $ Steiner points as vertices.

$\bullet$ We show that the construction in Gudmundsson et al. \cite{gudmundsson2007small} needs $\Omega (n \log n)$ points even for a convex point set and may not result in a planar graph.

\subsection{Organization}
In \cref{sec-pre}, we sketch the $\mathcal{ O}(n \log n)$ construction of Gudmundsson et al. \cite{gudmundsson2007small}. We prove that, even for convex point set, their construction  needs $\Omega (n \log n)$ points. We also show that their construction is not planar by considering a convex point set of 16 points for which their \pg~has a minor homeomorphic to $ K_{3,3}$.  In \cref{sec-cons}, we provide our construction  of $\mathcal{O}(n)$ size \planar~$ G $ for a convex point set $S$.

\section{$\mathsf{Manhattan~Network}$ for General Point Sets} \label{sec-pre}

For general point sets, Gudmundsson et al. \cite{gudmundsson2007small} proved the following theorem.

\begin{theorem}\cite{gudmundsson2007small}\label{theo_gud1}
	Let $ P $ be a set of $ n $ points. A \pgt~of $P$ consisting of $\Theta(n \log n)$ vertices and edges can be computed
	in  $ \mathcal{ O}(n \log n) $ time.
	
\end{theorem}

\begin{figure}[h!]
	\centering
	\includegraphics[scale=0.5]{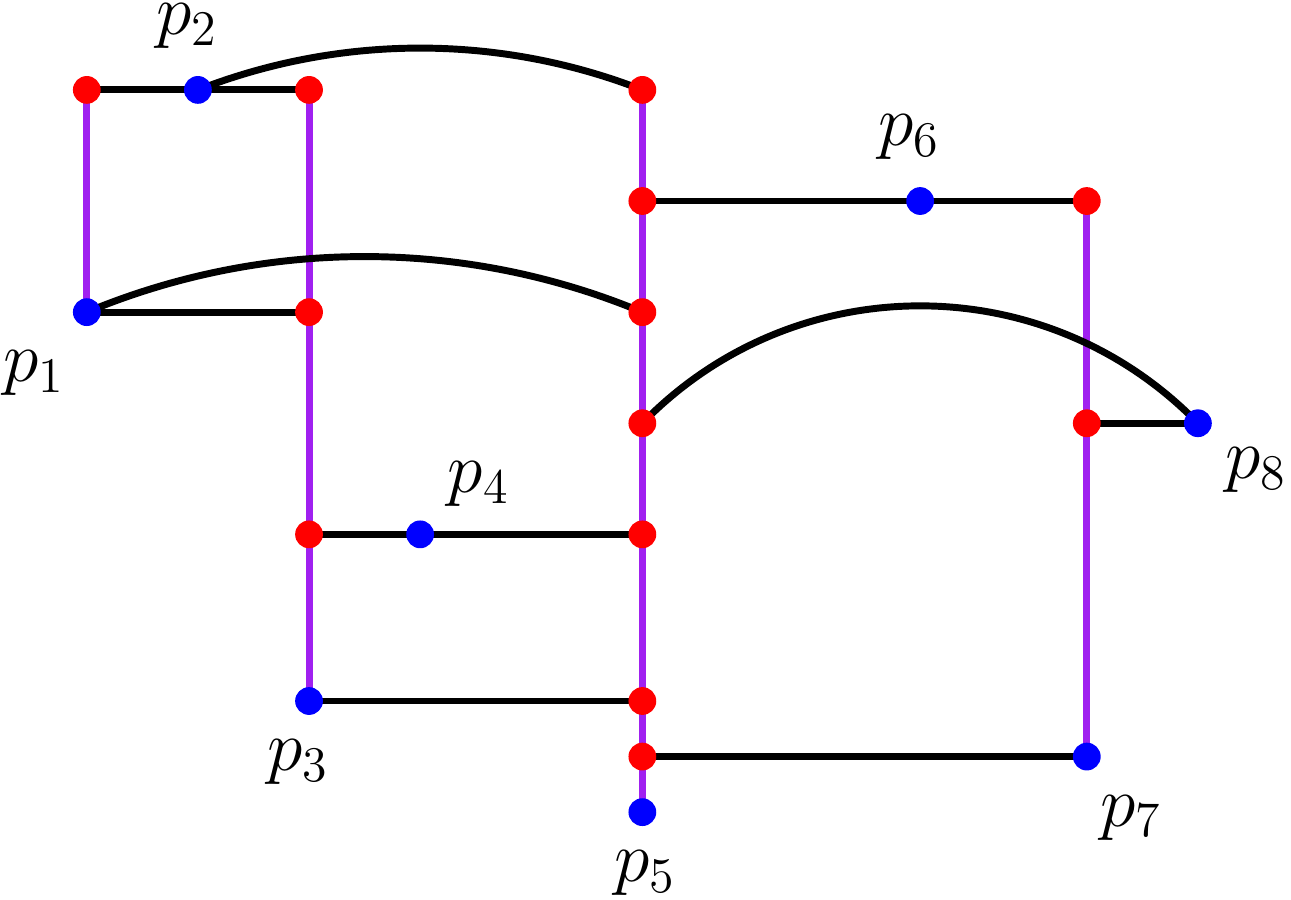}
	\caption{Construction of the \pgt~for $ S $. Points in $ S $ are in blue color and Steiner points are in red color. } 
	\label{fig-mn}
\end{figure}

Their construction is as follows: Sort the points in $P$ according to their $ x $-coordinate. Let $ m $ be the median $ x $-coordinate in $ P$. Then draw a vertical line $ L_m$ through $ (m,0) $. For each point $ p $ of $ S $, take an orthogonal projection on the line $ L_m $. Add Steiner points at each projection and join $p $ with its corresponding projection point. Then recursively do the same, on the $\frac{n}{2}$ points that have less $ x $-coordinate than $ p $ and $\frac{n}{2}$ points that have greater $ x $-coordinate than $ p $. Add a Steiner point at each projection.  \cref{fig-mn} illustrates the algorithm of Gudmundsson et al. \cite{gudmundsson2007small}.

Now we show that even for convex point set, this construction will need $\Omega (n \log n)$ Steiner points. In \cref{{fig-gen}}, for a set of sixteen points in convex position, we show that their network is not planar as it has a minor homeomorphic to $ K_{3,3}$ and the network uses 38 Steiner points.

\begin{figure}[ht!]
	\centering
	\subfigure[]
	{
		\includegraphics[scale=0.26]{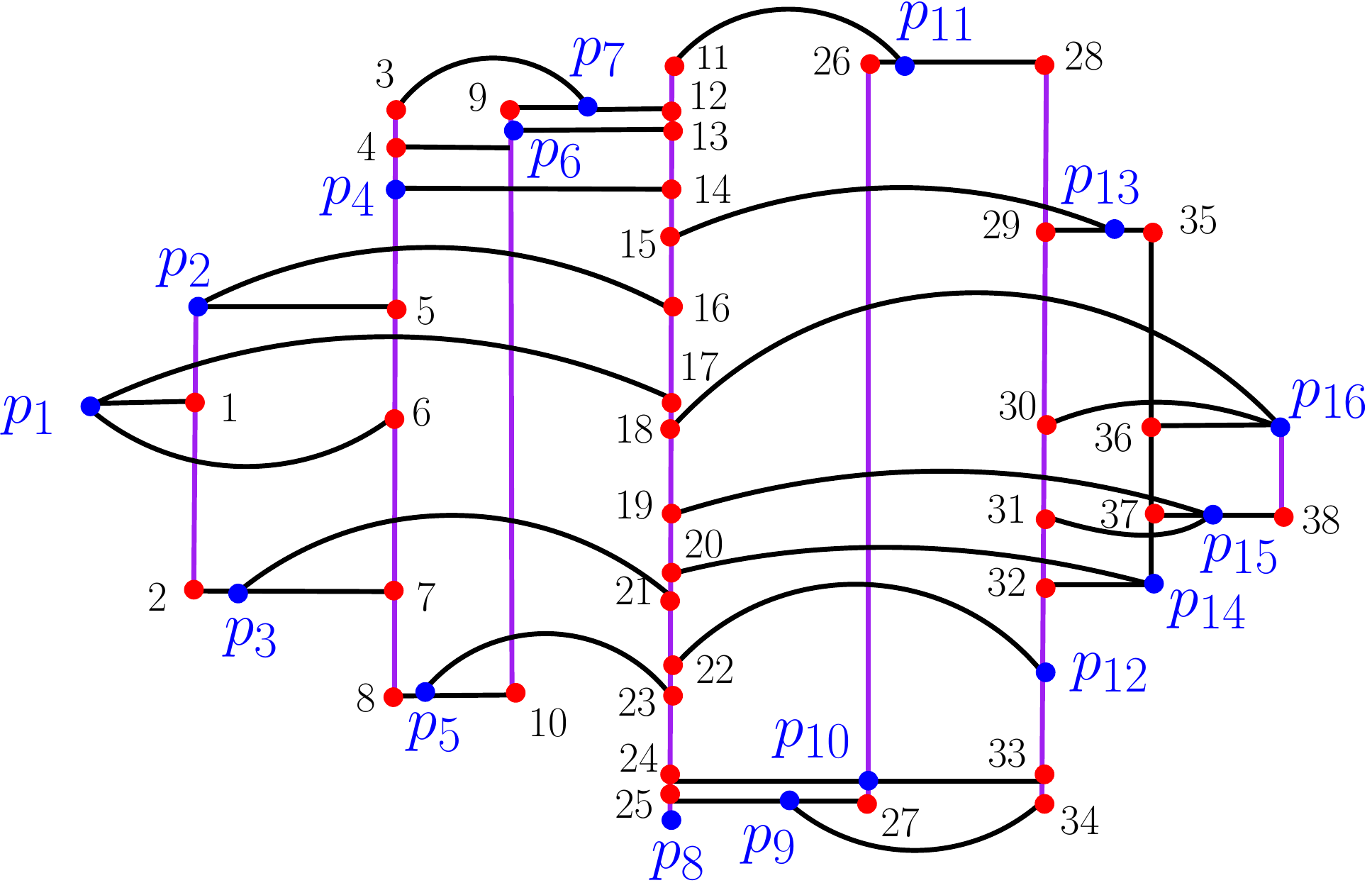}
		\label{fig-gen1}
	}
	\hspace{18mm}
	\subfigure[]
	{
		\includegraphics[scale=0.26]{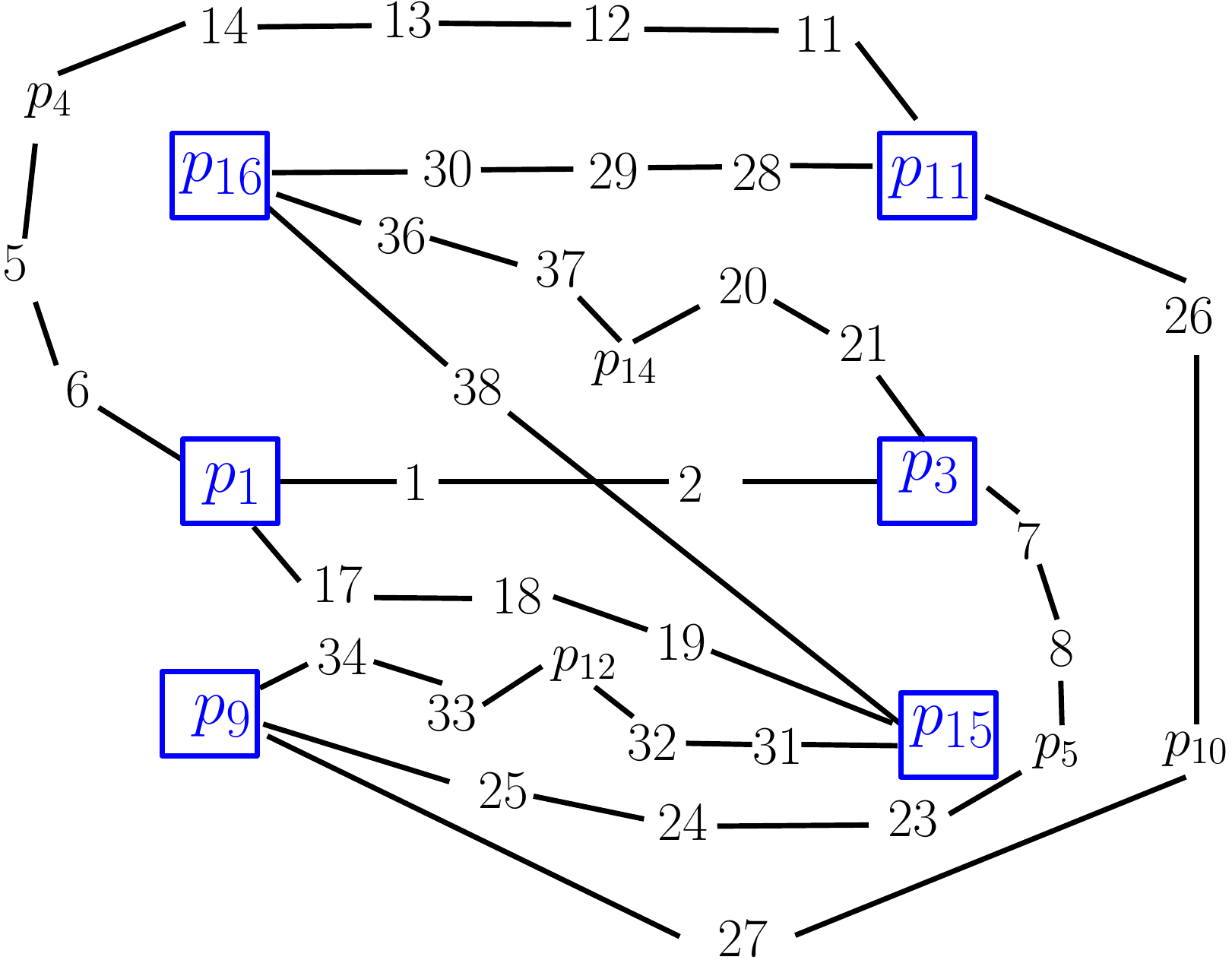}
		\label{fig-gen2}
	}
	
	\caption{(a) \pg~$G_A$ of a convex point set $A=\{p_1, p_2, \dots, p_{16}\}$ (blue color). Points colored in red  are Steiner points, and (b) $G'_A$, subgraph of $G_A$, that is homeomorphic to $K_{3,3}$.} 
	\label{fig-gen}
\end{figure}

\section{$\mathsf{Planar~Manhattan~Network}$ for a Convex Point Set}\label{sec-cons}

In this section, we construct a linear size \planar~$G $ for a convex point set $S$. $G$ uses $\mathcal{ O}(n)$ Steiner points and can be constructed in linear time. 
We organize this section as follows: After introducing some definitions and notations in \cref{sec-prel}, we construct a histogram partition $\mathcal{H}(\mathcal{OCP}(S))$ of an ortho-convex polygon $ \mathcal{OCP}(S) $  of the convex point set $ S $ in \cref{sec-ortho}.
In \cref{sec-graph} we construct our desired graph $ G=(V,E) $ where $ S \subseteq V $.	 In \cref{sec-path} we prove that $ G $ is a \pg~for $ S$. In \cref{sec-planar} we show that $ G $  is planar. In \cref{sec-con} we draw conclusions and state some interesting open problems.

\subsection{Preliminaries }\label{sec-prel}

A \emph{polygonal  chain},  with $ n $ vertices in the plane, is defined as an ordered set of vertices
$(v_1, v_2,\dots, v_n)$, such that any two consecutive vertices $v_i, v_{i+1}$ are connected by the line segment
$ \overline{v_i v_{i+1}} $, for $1\leq i<n$. It is said to be \emph{closed} when it divides the plane into two disjoint regions.  A \emph{polygon} is a bounded region which is enclosed by a closed polygonal chain in $ \mathbb{R}^2$. A line segment is \emph{orthogonal} if it is parallel either to the $x$-axis or $y$-axis.

\begin{definition}{\textbf{(Orthogonal polygon)}}
	A polygon is said to be an orthogonal polygon if all of its sides are orthogonal.
\end{definition}

\begin{definition}{\textbf{(Ortho-convex polygon)}}\cite{datta1990some}
	An orthogonal polygon $\mathcal{P}$ is said to be ortho-convex if every
	horizontal or vertical line segment connecting a pair of points in $\mathcal{P}$ lies totally within $\mathcal{P}$.
\end{definition}

\begin{definition}{\textbf{(Shortest $ L_1 $ path)}}
	A path between two points $ p $ and $ q $ is said to be a shortest $ L_1 $ path between them if the path consists of orthogonal line segments with total length $ \norm{pq }_1 $.
\end{definition}

\begin{lemma} \cite{chepoi2008rounding}
	For all pair of points in an ortho-convex polygon $\mathcal{P}$, there exist a shortest $ L_1 $ path between them in $\mathcal{P}$.
	
\end{lemma}

\subsection{$ \mathcal{OCP}(S) $ and $\mathcal{H}(\mathcal{OCP}(S))$ }\label{sec-ortho}

Let $ S=\{p_{1}, p_{2}, \dots, p_{n}\} $ be a convex point set of size $ n $ in $ \mathbb{R}^2 $. For any point $ p\in S $, let $x(p)$ and $ y(p)$ be its  $ x $ and $ y $-coordinate, respectively. We assume that the points in $ S $ are ordered with respect to an anticlockwise orientation along their convex hull. Without loss of generality let this ordering be $ p_{1}, p_{2}, \dots, p_{n} $  and also we assume that $ p_{1} $ is the top most point in $ S$, i.e., point having the largest $ y $-coordinate  in $ S $(for multiple points having largest $ y $-coordinate, we take the one that has smallest $x$-coordinate).   We denote the right most point of $ S $ as $r$. Analogously, let $l$, $t$, and $b$ denote the left most, the top most and the bottom most point of $ S $, respectively. So $ t=p_{1} $.
We will consider the point set for the case that $x(p_1) < x(b)$. For the case of $x(p_1) \geqslant x(b)$, both the construction and the proof are symmetric (by taking the mirror image of the point set  with respect to the line $y=y(p_1)+1$).

 \vspace{2mm}
A polygonal chain is said to be a $ xy$-monotone if any orthogonal line segment intersects the chain in a connected set. Now we will construct an ortho-convex polygon  $ \mathcal{OCP}(S) $, where points in $S$ lie on the boundary of  $ \mathcal{OCP}(S)$. $ \mathcal{OCP}(S) $ consists of four $ xy$-monotone chains. Let us denote these chains as $ C_{rt},~C_{tl}, ~C_{lb}, ~\text{and}~C_{br} $. $ C_{rt}  $ defines a $ xy$-monotone chain with the endpoints at $ r ~\text{and}~t $. Analogously, $C_{tl}, ~C_{lb}, ~\text{and}~C_{br}   $ are defined. While constructing the chain $C_{rt}$, we do the following: For any pair of consecutive points $p, q$, if $x(p) > x(q)$ then we draw two line segments $\overline{pp'}, \overline{qp'}$, where $p'=(x(p), y(q))$, else we extend the chain upto the next point. In \cref{algo_1.1} and  \cref{algo_1.2} ,  we describe the construction of $ C_{rt} $ and $ C_{br} $ respectively. Construction for the  all other monotone chains follows the same set of rules.  See  \cref{fig_chain1} for an illustration.

\vspace{2mm}

\begin{breakablealgorithm}
	\caption{Construction of the chain $ C_{rt} $}\label{algo_1.1}
	\hspace*{-2cm} \textbf{Input:} A set of $ k $ points $ p_{i}(=r), p_{i+1}, \dots, p_{i+k-1}(=t)  $ such that $   x(p_{j+1}) \leqslant  x(p_{j}) , y(p_{j+1}) \geqslant y(p_{j})  $ for $ i \leqslant j < (i+k-1) $\\
	\hspace*{-8.35cm} \textbf{Output:} The chain $ C_{rt} $
	\begin{algorithmic}[1]
		\For{$j=i~ \text{to}~ (i+k-2) $}
		\If{$x(p_j)=x(p_{j+1})$ or $y(p_j)=y(p_{j+1})$ }
		\State Join the  line segments $\overline{p_{j}~p_{j+1}}$ 
		\Else
		\State Create a Steiner point $ p_{j,j+1} = (x(p_{j}), y(p_{j+1}))$
		\State Join the  line segments $\overline{p_{j}~p_{j,j+1}}$ and $\overline{p_{j,j+1}~p_{j+1}}$
		\EndIf
		\EndFor
%
%

	\end{algorithmic}
\end{breakablealgorithm}

\newpage

\begin{figure}[ht!]
	\centering
	\subfigure[]
	{
		\includegraphics[scale=0.36]{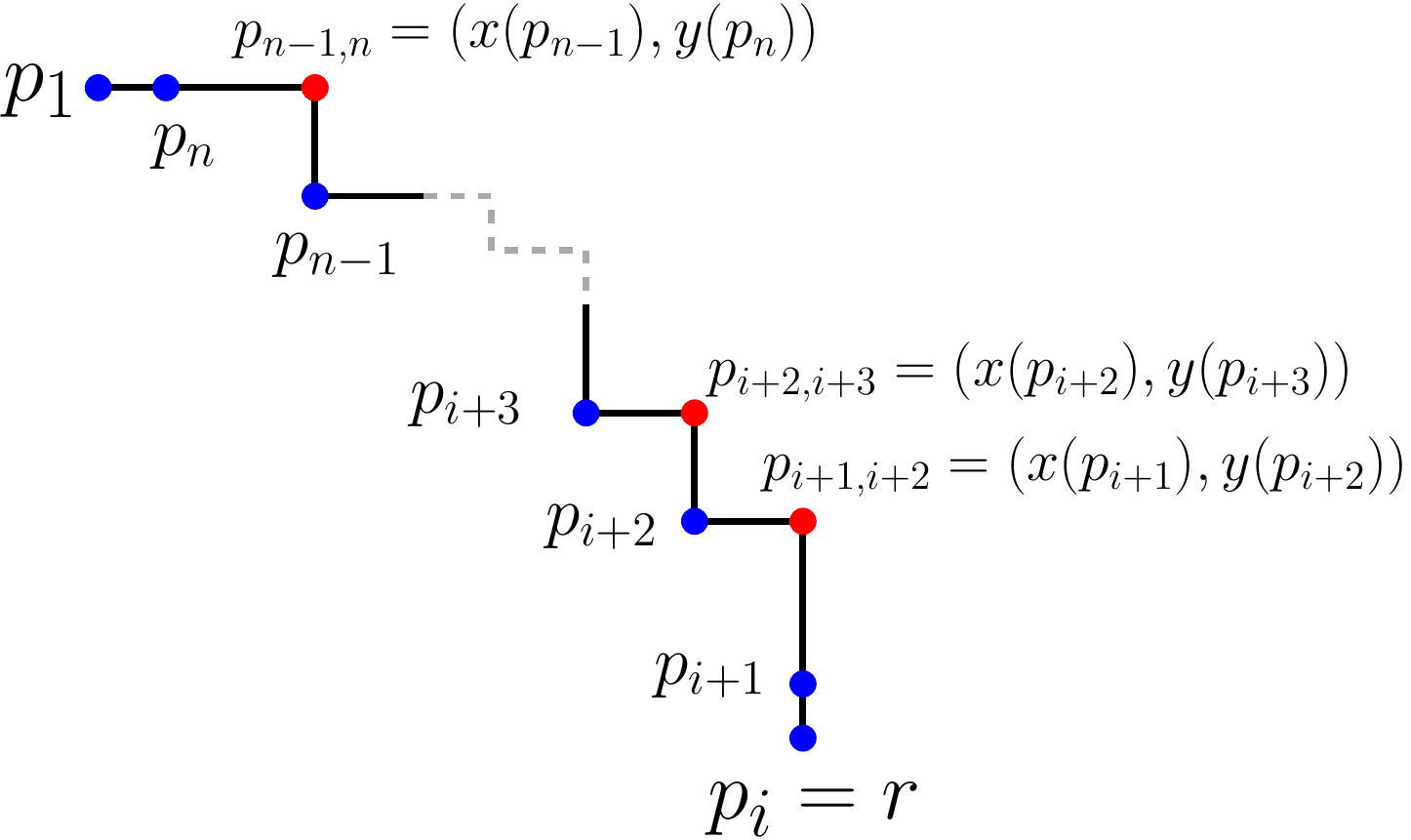}
		\label{fig-11}
	}
	\subfigure[ ]
	{
		\includegraphics[scale=0.36]{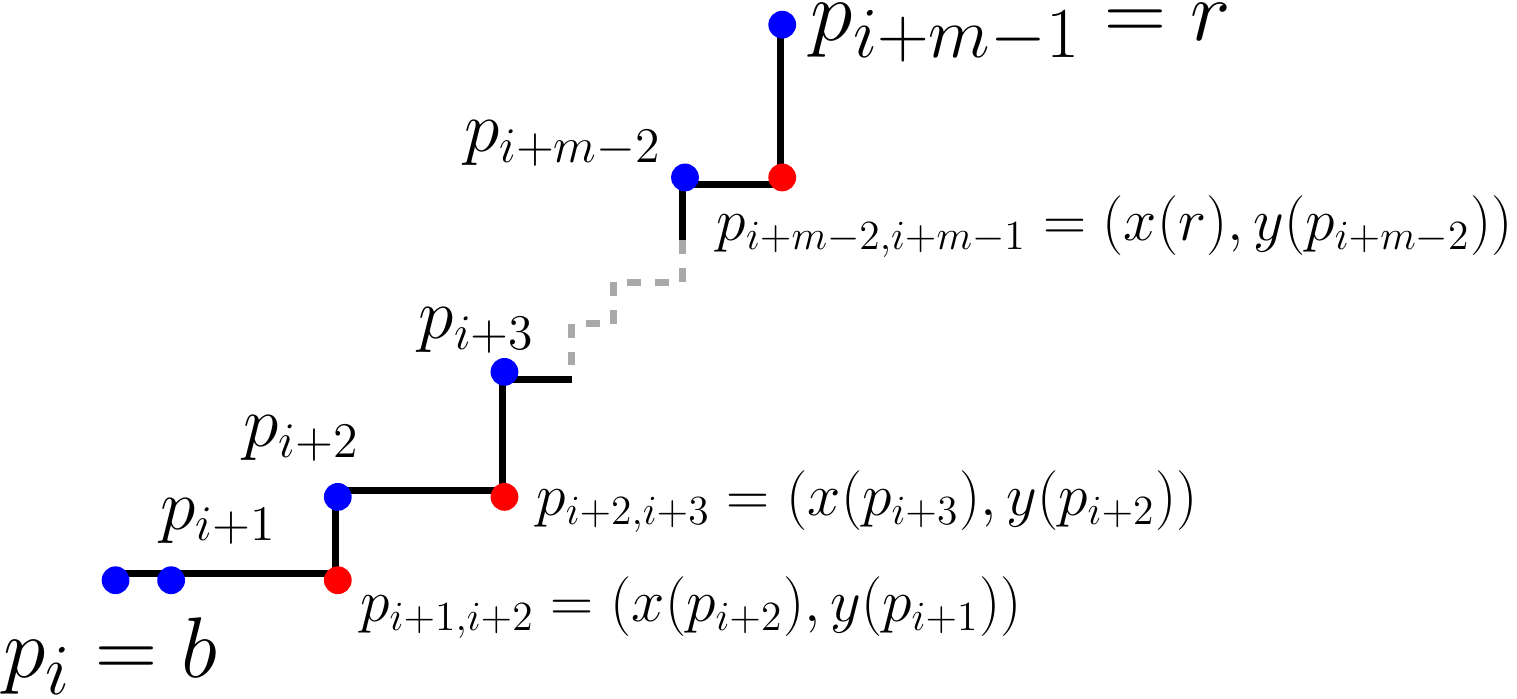}
		\label{fig-22}
	}

	\caption[Optional caption for list of figures]{Construction of chains (a) $ C_{rt}$ and (b) $ C_{br} $ from a given convex point set (blue color) }
	\label{fig_chain1}
\end{figure}

\begin{breakablealgorithm}
	\caption{Construction of the chain $ C_{br} $}\label{algo_1.2}
	\hspace*{-2cm} \textbf{Input:} A set of $ m $ points $ p_{i}(=b), p_{i+1}, \dots, p_{i+m-1}(=r)  $ such that $ x(p_{j+1}) \geqslant x(p_{j}), y(p_{j+1}) \geqslant y(p_{j})  $ for $ i \leq j < (i+m-1) $\\
	\hspace*{-8.5cm} \textbf{Output:} The chain $ C_{br} $
	\begin{algorithmic}[1]
		
			\For{$j=i~ \text{to}~ (i+m-2) $}
		\If{$x(p_j)=x(p_{j+1})$ or $y(p_j)=y(p_{j+1})$ }
		\State Join the  line segments $\overline{p_{j}~p_{j+1}}$ 
		\Else
		\State Create a Steiner point $ p_{j,j+1} = (x(p_{j+1}), y(p_{j}))$
		\State Join the  line segments$\overline{p_{j}~p_{j,j+1}}$ and $\overline{p_{j,j+1}~p_{j+1}}$
		\EndIf
		\EndFor

		
	\end{algorithmic}
\end{breakablealgorithm}

In \cref{fig_chain}, we illustrate an example of a convex point set $ S $ of size 15 and the ortho-convex polygon $ \mathcal{OCP}(S) $ is shown in \cref{fig-44}.

\begin{figure}[ht!]
	\centering
		\subfigure[]
	{
		\includegraphics[scale=0.4]{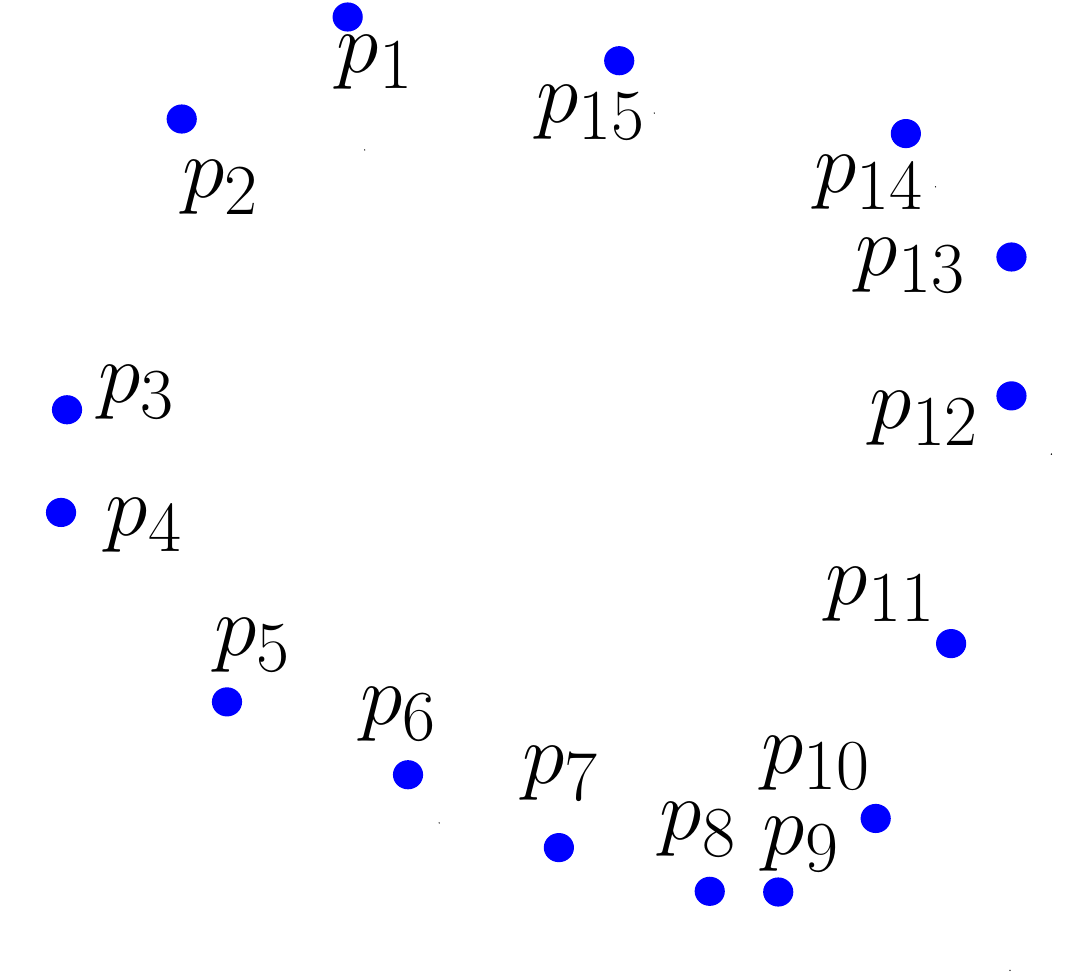}
		\label{fig-33}
	}
	\subfigure[ ]
	{
		\includegraphics[scale=0.4]{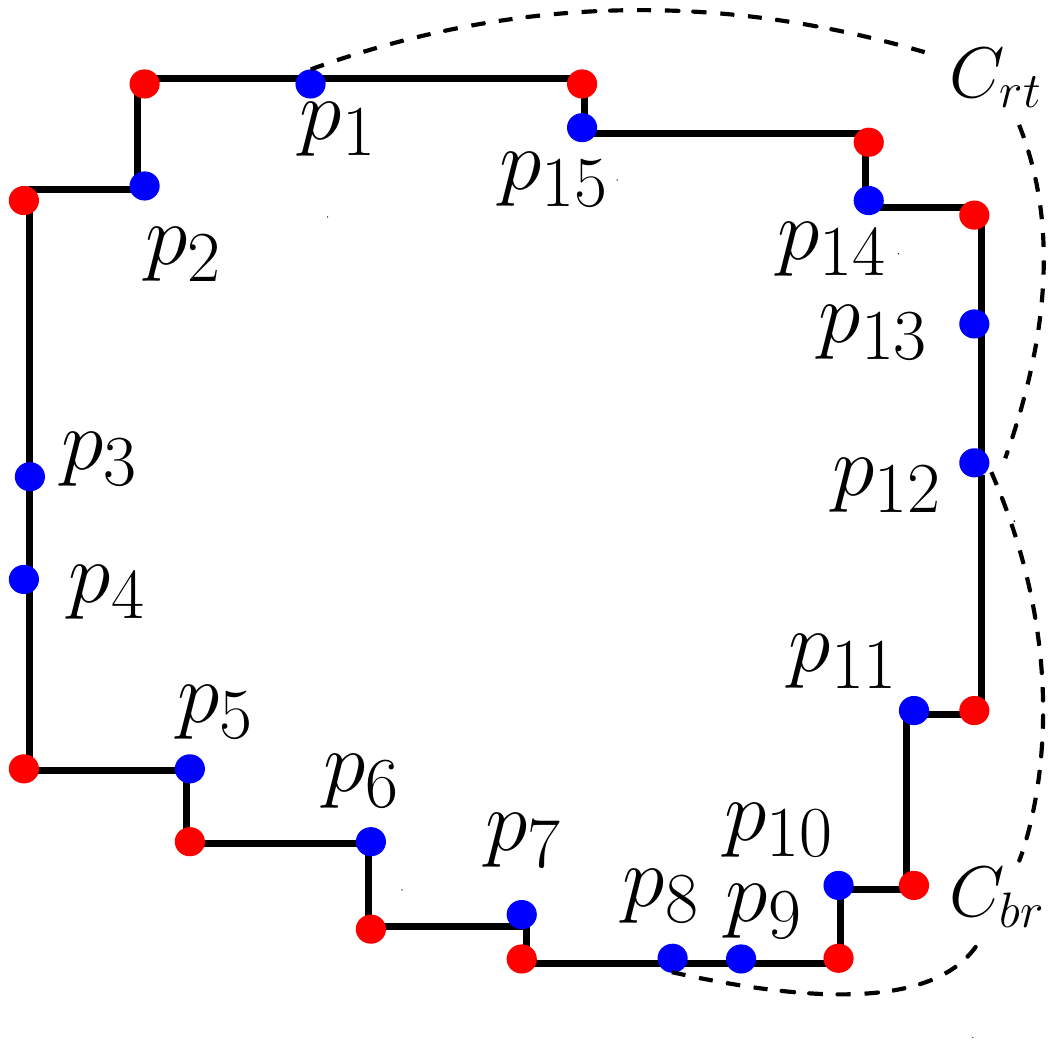}
		\label{fig-44}
	}
	
	\caption[Optional caption for list of figures]{(a) Example of a set $ S $ of 12 points in convex position, (b) $ \mathcal{OCP}(S) $ of $ S $. }
	\label{fig_chain}
\end{figure}

\begin{definition}{\textbf{(Histogram)}}
	A histogram $ H $ is an orthogonal polygon consisting of a boundary edge $ e $, called as its base, such that for any point $ p \in H $, there exists a point $ q\in e  $ such that the line segment $\overline{pq}$ is orthogonal and it lies completely in $ H $. 
	
\end{definition}

If the base is horizontal (respectively, vertical) we say it is a $horizontal$ (respectively, $vertical$) histogram. If its interior is above the base it is called an $upper$ histogram. Similarly, we can define the $lower$, $left$, and $right$ histograms. 
 Now we construct a histogram partition $\mathcal{H}(\mathcal{OCP}(S))$ of $ \mathcal{OCP}(S)$.

Let $ L =\overline{pq}$ be a vertical line segment such that both the points $ p $ and $ q$ are on the boundary of $\mathcal{OCP}(S)$. We define  $H_{L}^{r}$ and  $H_{L}^{l}$ to denote a right-vertical and left-vertical histogram, respectively, with base $L=\overline{pq}$. Similarly, for a horizontal line segment $L'=\overline{p'q'}$, where both the points $ p' $ and $ q' $ are on the boundary of $\mathcal{OCP}(S)$,  we define $H_{L'}^{u}$ and   $H_{L'}^{b}$  to denote an upper-horizontal and lower-horizontal histograms, respectively, with base $L'=\overline{p'q'}$. Let $\mathtt{proj}_L(p)$ be the orthogonal projection of the point $ p $ on the line containing the segment $ L $. For a set $ A $ of orthogonal line segments and a point set $ S $, we say $ A $ can see $ S $ if $\forall p \in S $ there is at least one line segment $ L \in A $ such that $ \mathtt{proj}_L(p) \in L $. For a vertical (respectively, horizontal) line segment $L$, we define $x(L)$  (respectively, $y(L)$) to be the $x$-coordinate (respectively, $y$-coordinate) of $ L $.\par 
We obtain a histogram partition $\mathcal{H}(\mathcal{OCP}(S))$  of $ \mathcal{OCP}(S) $ by recursively drawing vertical and horizontal lines as follows (see \cref{fig-hist}):

\begin{description}
	\item[Step 1] Let $ q_1~(\in C_{lb})$ be the  intersection point of the boundary of $ \mathcal{OCP}(S) $ with the vertical line containing $ p_1 $. First, we draw a vertical line segment $L_{1}=\overline{p_{1}q_1} $. We define two sets $S(H_{L_{1}}^{l})$ and $ S(H_{L_{1}}^{r}) $ such that $S(H_{L_{1}}^{l})= \{q \in S \colon y(t) \geq y(q) \geq y(q_1)$ and $x(q) \leq x(q_1)\}$, $S(H_{L_{1}}^{r})= \{q \in S \colon y(t) \geq y(q) \geq y(q_1)$ and $x(q) \geq x(q_1)\}$. In this step, we construct two vertical histograms $ H_{L_{1}}^{l} $ and  $ H_{L_{1}}^{r} $. If $S(H_{L_{1}}^{l}) \cup S(H_{L_{1}}^{r})=S$, i.e., $ L_1 $ can see $ S $ we stop, else we proceed to Step 2.
	
	\item[Step 2:] 
	Let $ q_2~( \notin C_{lb})$ be  the intersection point of the boundary of $ \mathcal{OCP}(S) $ with the horizontal line containing $ q_1 $. Then we draw a horizontal line segment $L_{2}=\overline{q_{1}q_2} $. Here we define the set $S(H_{L_{2}}^{b})= \{z \in S \colon x(q_1) \leq x(z) \leq x(q_2)$ and $ y(z) \leq y(q_2)\}$. In this step, we construct the lower histogram $ H_{L_{2}}^{b} $ with base $ L_2 $. If $S(H_{L_{1}}^{l}) \cup S(H_{L_{1}}^{r}) \cup S(H_{L_{2}}^{b})=S$, i.e., $ \{L_1, L_2\} $ can see $ S $ we stop, else we proceed to the next step.
	
	\item[Step 3:] 
	Let $ q_3~( \notin C_{rt})$ be  the intersection point of the boundary of $ \mathcal{OCP}(S) $ with the vertical line containing $ q_2 $. Then we draw a vertical line segment $L_{3}=\overline{q_{2}q_3} $. Here we define the set $S(H_{L_{3}}^{r})= \{w \in S \colon y(q_2) \geq y(w) \geq y(q_3)$ and $x(q) \geq x(q_3)\}$. In this step, we construct the right histogram $ H_{L_{3}}^{r} $ with base $ L_3 $. If $S(H_{L_{1}}^{l}) \cup S(H_{L_{1}}^{r}) \cup S(H_{L_{2}}^{b}) \cup S(H_{L_{3}}^{r})=S$, i.e., $ \{L_1, L_2, L_3\} $ can see $ S $ we stop, else we proceed in the similar manner.

	We assume that  this process terminates after $ k$ steps, and we obtain a set $\mathcal{L}$ of orthogonal line segments $ \{L_1, L_2, \dots L_k\} $ for some $ k \in \mathbb{N}$ such that $ \{L_1, L_2, \dots, L_k\} $ can see $ S $. In this process, we add $ k $ Steiner points $ \{q_i \colon 1\leq i \leq k\} $. Each $ q_i $ belongs to the boundary of  $ \mathcal{OCP}(S) $.  
	
\end{description}

The process terminates in one of the four following  configurations which are based on the position of the points $ b $ and $ r $ (see \cref{fig_hist}).

\begin{description}
	\item[Type-1 ] $ L_k $ is vertical and  $\mathtt{proj}_{L_{k-1}}(b) \in L_{k-1} $, i.e., $ L_{k-1} $ sees $ b$. 
	\item[Type-2 ] $ L_k $ is vertical and $\mathtt{proj}_{L_{k-1}}(b) \notin L_{k-1} $. 
	\item[Type-3 ] $ L_k $ is horizontal and $ \mathtt{proj}_{L_{k-1}}(r) \in L_{k-1} $, i.e., $ L_{k-1} $ sees $ r$. 
	\item[Type-4 ] $ L_k $ is horizontal and $ \mathtt{proj}_{L_{k-1}}(r) \notin L_{k-1} $.
	
\end{description}

From now onwards, we assume that  $ L_1, L_2, \dots L_k $ are the segments inserted in $ \mathcal{OCP}(S) $ while constructing $\mathcal{H}(\mathcal{OCP}(S))$. Let $ \mathcal{L}= \cup_{i=1}^n L_i $. So for  any point $ p\in S $, there is at least one line segment $ L \in  \mathcal{L} $ such that  $\mathtt{proj}_L(p) \in L $ and the segment $\overline{p~\mathtt{proj}_L(p)} $ completely lies in $ \mathcal{OCP}(S) $. 

\begin{figure}[ht]
	\centering
	\subfigure[Type 1]
	{		
		\includegraphics[scale=0.35]{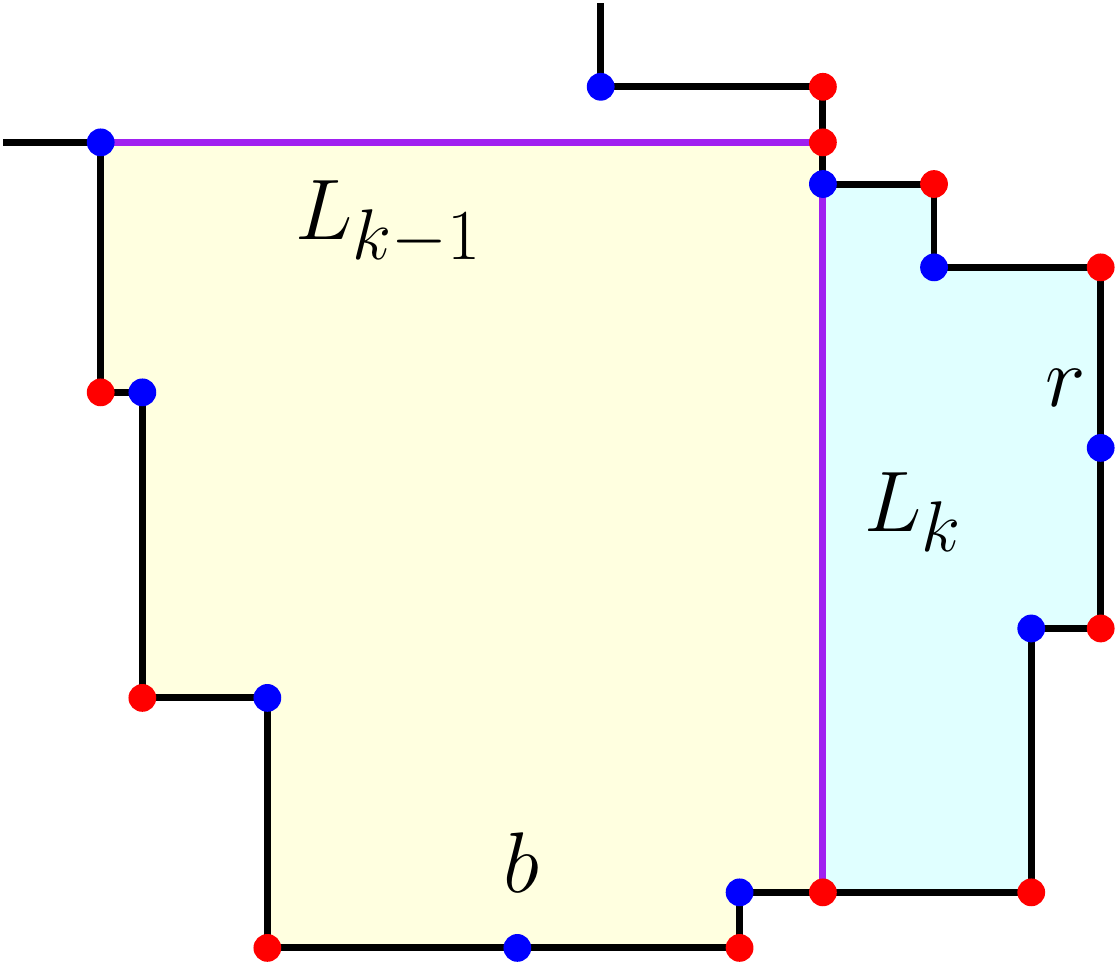}
		\label{}
	}
	\subfigure[ Type 2]
	{
		\includegraphics[scale=0.35]{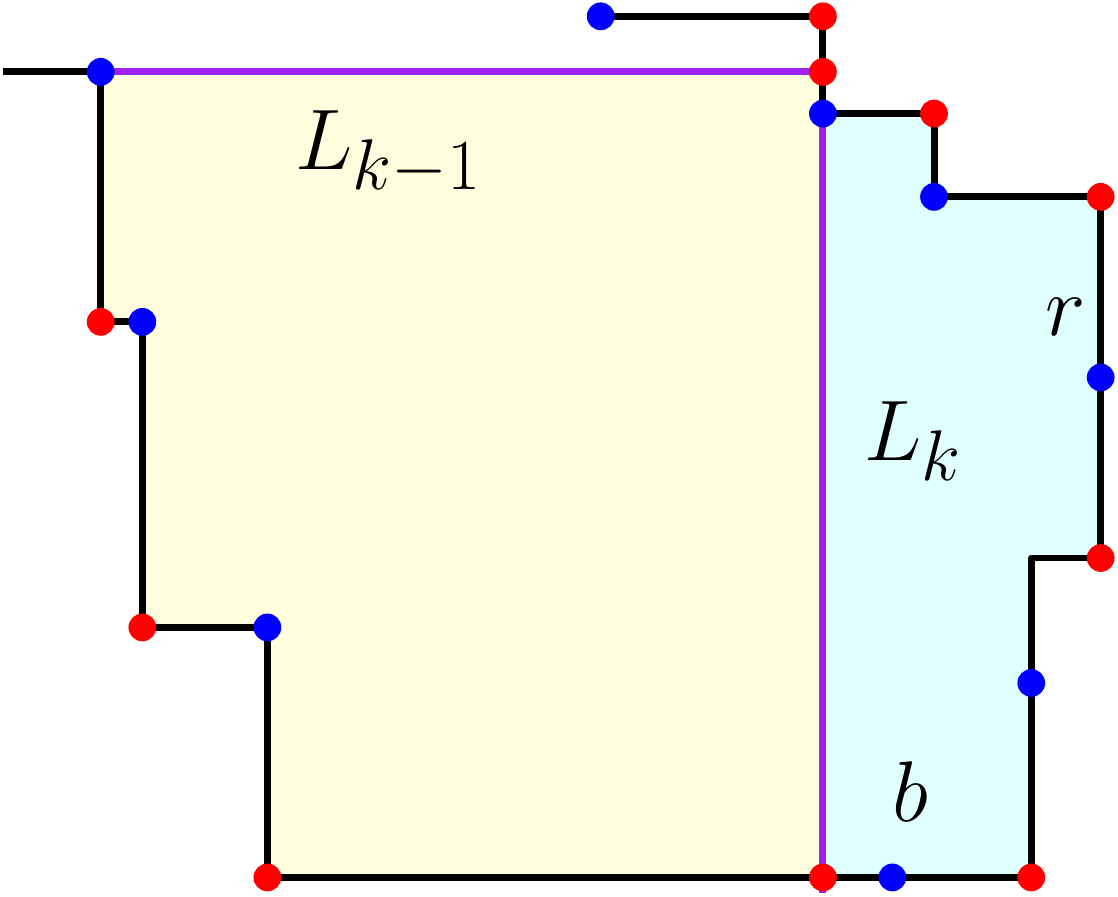}
		\label{}
	}
	
	\subfigure[ Type 3]
	{
		\includegraphics[scale=0.35]{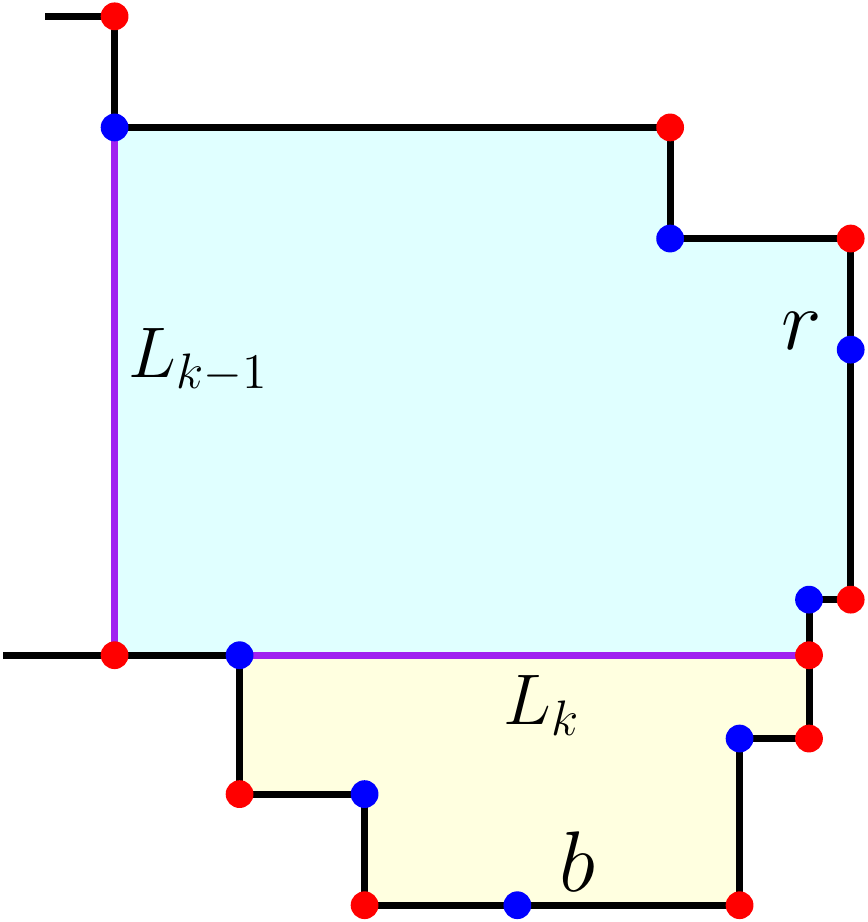}
		\label{}
	}
	\hspace{15mm}
	\subfigure[Type 4]
	{
		\includegraphics[scale=0.35]{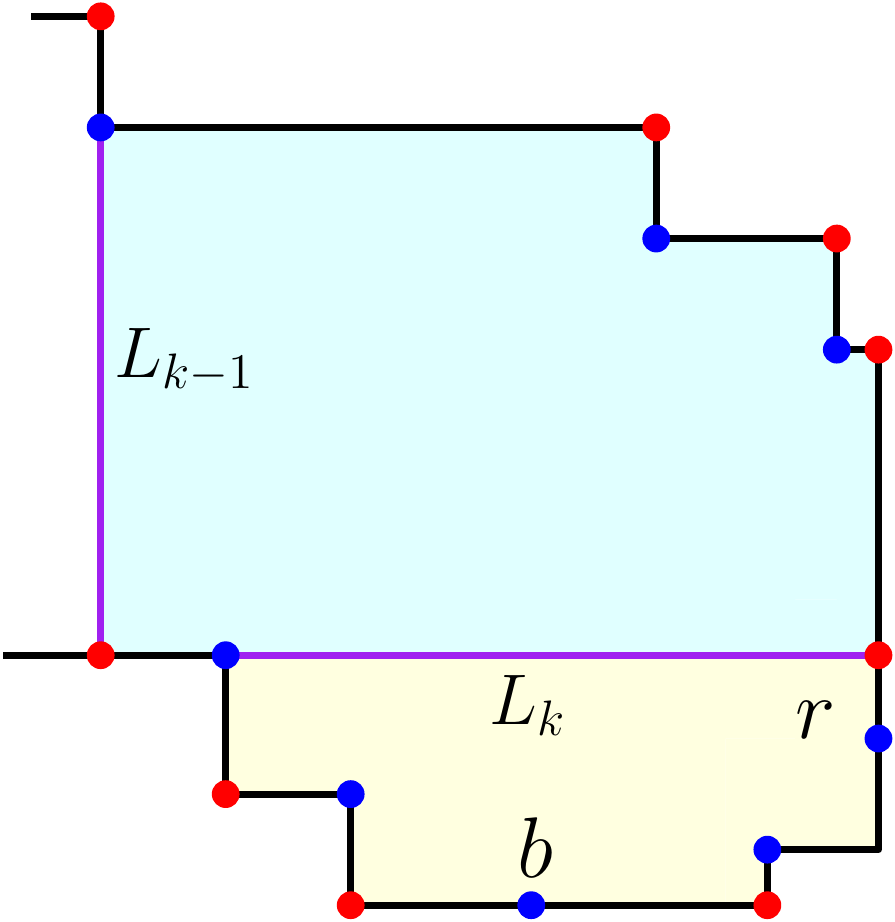}
		\label{}
	}
	
	\caption[[Optional caption for list of figures]{Types of the histogram containing $ b $ and $ r $ in $ \mathcal{OCP}(S) $.}
	\label{fig_hist}
\end{figure}

\begin{lemma}\label{lem-L}
	$ \mathcal{H}(\mathcal{OCP}(S)) $ can be constructed in linear time. 
\end{lemma}

\begin{proof} 
	Let $L_i (S) = \{p \in S \colon L_i~\text{can see}~p\}$. First we show that $ \mathcal{H}(\mathcal{OCP}(S)) $ is    a histogram partition in $\mathcal{OCP}(S)$, i.e., $\cup^{k}_{i=1}L_i (S)= S$. $L_1$ sees all points $q \in S$ having the property that $ y(q_1) \leq y(q) \leq y(t)$ as $\mathcal{OCP}(S)$ is an ortho-convex polygon and these points are part of $ xy $-monotone chains $\{C_{rt}, C_{tl}, C_{lb}, C_{br}\}$. So $L_1(S)$ consists of all the points in $S$ that lie above $L_2$. Moreover, all the points above $L_2$ are part of the histogram defined by the base $L_1$. Now our concern is only about the points of $ S $ that are below $L_2$.
	Now $L_2$ can see the points $q \in (S \setminus L_1(S))$ having the property that $ x(q_1) \leq x(q) \leq x(q_2)$.  These points are part of the histogram with the base $ L_2 $. Now we can apply the same argument inductively. This leads to the claim that $\cup^{k}_{i=1}L_i (S)= S$, i.e., $ \mathcal{L}=\{L_1, \dots L_k\} $ can see $ S $. 
	Observe that the segments in  $ \mathcal{L}$ can be computed by walking around the boundary of $ \mathcal{OCP}(S)$ in linear time. Hence, $ \mathcal{H}(\mathcal{OCP}(S)) $ can be constructed in linear time.
\end{proof}

\begin{figure}[h!]
	\centering
	\includegraphics[scale=0.26]{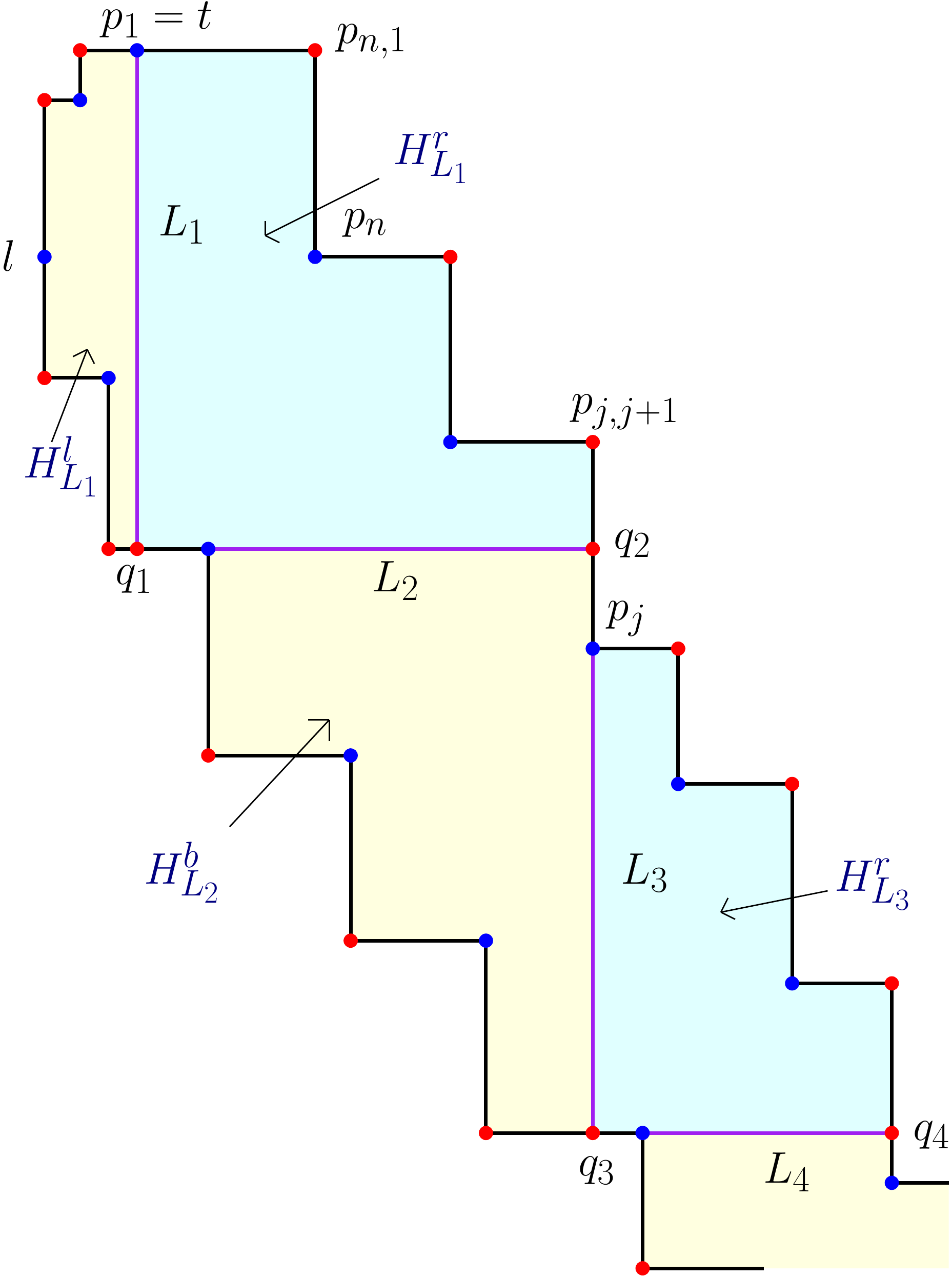}
	\caption{$ \mathcal{H}(\mathcal{OCP}(S))$ of a convex point set $S$  } 
	\label{fig-hist}
\end{figure}

\subsection{Construction of Planar Manhattan Network }\label{sec-graph}

Now we describe our construction of planar Manhattan network $G=(V,E) $ for a convex point set $ S $.  For an illustration of the steps of \cref{algo_2}, see \cref{fig_step}. Recall that $\mathtt{proj}_L(p)$ denotes the orthogonal projection of the point $ p $ on the line containing the segment $ L $ and  $q(H)$ denotes the histogram containing $q \in S$ in $\mathcal{H}(\mathcal{OCP}(S))$. Let $e_1, e_2, e_3$ be the bases of $l(H), b(H), $ and $r(H)$, respectively, where $l(H)$ (respectively $b(H)$ and $r(H)$) denotes the histogram containing $l$ (respectively $b$ and $r$) of $S$. First, we draw the segments $e'_1= \overline{l~\mathtt{proj}_{e_1}(l)}, e'_2= \overline{b~\mathtt{proj}_{e_2}(b)},$ and $e'_3= \overline{r~\mathtt{proj}_{e_3}(r)}$ in $\mathcal{OCP}(S)$. Let $\mathcal{L'}= \mathcal{L} \cup \{ e'_{1}, e'_{2},e'_{3} \}$. Next, for each $q \in S$, if $\mathtt{proj}_L(q) \in L$ where $L \in \mathcal{L'} \cap q(H) $, we draw the line segment $\overline{q~\mathtt{proj}_L(q)}$ in $\mathcal{OCP}(S)$. Then if both $H^r_{L_k} $ and $ e_2 $ exist, we draw the segments 
$\mathtt{proj}_{e_2}(z)$, for each point $ z \in S \cap H^r_{L_k} $. Also if both $H^b_{L_k} $ and $ e_3 $ exist, we draw the segments $\mathtt{proj}_{e_3}(w)$, for each point $ w \in S \cap H^b_{L_k} $. In this process, all the line segments we join, we add them into edges of $T$. Also all the extra points we created to make an orthogonal projection, we add them into the set $ T $ of Steiner vertices. Our algorithm ends with removing some  specific line segments, that is stated in the Steps 28-32 in \cref{algo_2}.  We illustrate this algorithm in \cref{fig-graph}.

\begin{breakablealgorithm}
	\caption{ Construction of $G= (V=S \cup T, E$)}\label{algo_2}
	\hspace*{-15mm}	 \textbf{Input:} $\mathcal{H}(\mathcal{OCP}(S))$ of a convex point set $ S =\{p_{1}(=t), p_{2}, \dots, p_{n}\}$.\\{{\color{blue} Let  $ \{L_1, L_2, \dots L_k\} $ be the segments and $ \{q_i \colon 1 \leq i \leq k\} $  be the set of \\ points inserted in $ \mathcal{OCP}(S) $ during the construction of $\mathcal{H}(\mathcal{OCP}(S))$. }}\\
	
	\hspace*{-21mm} \textbf{Output:} A \planar~$ G=(V=S\cup T,E) $ of $ S $.
	\begin{algorithmic}[1]
		
		\State $S \gets  \{~p_{i} \colon 1 \leq i \leq n\};$
		\State $T \gets  \{~p_{i, i+1} \colon~~ 1 \leq i \leq n\} \cup \{~q_{i} \colon~~ 1 \leq i \leq k\};$ 
		\State $E \gets \{\overline{p_{i}p_{i, i+1}} \colon 1 \leq i \leq (n-1)\} \cup \{\overline{p_{i+1}p_{i, i+1}}\colon 1 \leq i \leq (n-1)\} \cup \overline{p_{n}p_{n,1}} \cup \overline{p_{1}p_{n,1}};$ \Comment{see \cref{fig-s3}}
		\State Draw the line segments (if they do not exist) $e'_1= \overline{l~\mathtt{proj}_{e_1}(l)}, e'_2= \overline{b~\mathtt{proj}_{e_2}(b)},$ and $e'_3= \overline{r~\mathtt{proj}_{e_3}(r)}$ \Comment{ {\color{blue}$e_1, e_2,$ and $ e_{3} $  are the bases of the histograms $l(H),b(H),$ and $r(H)$, respectively.}}
		\State $ T=T \cup \{\mathtt{proj}_{e_1}(l), \mathtt{proj}_{e_2}(b), \mathtt{proj}_{e_3}(r)\} $
		
		\State $ \mathcal{L'}=\{L_1, L_2, \dots, L_{k}, e'_{1}, e'_{2},e'_{3} \} $ 
		
		\For{\textit{each point $ q\in S $}}
		\For{\textit{each line $ L \in \mathcal{L'} \cap q(H) $}}
		\If{$\mathtt{proj}_{L}(q) \in L$}
		\State $T=T \cup \mathtt{proj}_{L}(q);$ 
		\State $E= E \cup \overline {q~\mathtt{proj}_{L}(q)}$\Comment{{\color{blue} {see}} \cref{fig-s8}}
		\EndIf
		\EndFor
		\EndFor

		\If{Both $H^r_{L_k} $ and $ e_2 $ exist}
		\For{\textit{each point $ z \in S \cap H^r_{L_k} $}}
		\State $T=T \cup \mathtt{proj}_{e_2}(z);$
		\State $E= E \cup \overline {z~\mathtt{proj}_{e_2}(z)}$\Comment{{\color{blue} {see}} \cref{fig-s15}}
		\EndFor
		\EndIf
		
		\If{Both $H^b_{L_k} $ and $ e_3 $ exist}
		\For{\textit{each point $ w \in S \cap H^b_{L_k} $}}
		\State $T=T \cup \mathtt{proj}_{e_3}(w);  $
		\State $E= E \cup \overline {w~\mathtt{proj}_{e_3}(w)}$\Comment{{\color{blue} {see}} \cref{fig-s19} }
		\EndFor
		\EndIf

		\For{\textit{each horizontal line segment $L \in  \mathcal{L'}$}}
		\State Let $ L $ contains $ k_1 $ vertices $ a_{1}, a_2, \dots, a_{k_1} $, where $ x(a_{i})< x(a_{i+1})$ for $ 1 \leq i < k_1 $
		\For{\texttt{$1\leq i \leq (k_1-1)$}}
		\State $E= E \cup \overline{a_{i}a_{i+1}}$\Comment{{\color{blue} {see}} \cref{fig-s14} }
		\EndFor		
		\EndFor

		\For{\textit{each vertical line segment $L \in  \mathcal{L'}$}}
		\State Let $ L $ contains $ k_2 $ vertices $ b_{1}, b_{2}, \dots, b_{k_2} $, where $ y(b_{i})< y(b_{i+1}) $ for $ 1 \leq i < k_2 $
		\For{\textit{$1\leq i \leq (k_2-1)$}}
		\State $E= E \cup \overline{b_{i}b_{i+1}} $\Comment{{\color{blue} {see}} \cref{fig-s16}}
		\EndFor		
		\EndFor
		
		\State	Delete the following three edges if they exist. 
		
		\State{$(i)$ The edge  $ (u_1, \mathtt{proj}_{e_1}(l)) $ on the line $ e'_1 $ provided that $ u_1 \neq l $.\Comment{{\color{blue} {see}} \cref{fig-s191}}}
		
		\State	$(ii)$  For the {\bf Types 1 or 4}, the edge  $ (u_2, \mathtt{proj}_{e_2}(b)) $ on the line $ e'_2 $  provided that $ u_2 \neq b $.\Comment{{\color{blue} {see}} \cref{fig-s193}}
		
		\State	$(iii)$  For the {\bf Types 2 or 3}, the edge $ (u_3, \mathtt{proj}_{e_3}(r)) $ on the line $ e'_3 $   provided that $ u_3 \neq r $. \Comment{{\color{blue} {see}} \cref{fig-s192}}
		
		\State	For the {\bf Types 1 or 3}, delete all the vertices $ v $ on the line $ L_k $ where $ v \notin \{ \mathtt{proj}_{e_3}(r), \mathtt{proj}_{e_2}(b)\}  $ and $ v $ is not a point on the boundary of $\mathcal{OCP}(S)$.
		
		\State \textbf{return} $G=(S \cup T, E)$

	\end{algorithmic}
\end{breakablealgorithm}

 Notice that for each point in $S $, \cref{algo_2} adds at most three Steiner vertices in $ G $. Specifically, $ |V(G)| \leq 4n $ and $ |E(G)| \leq 5n $. So  both the number of vertices and edges in $ G $ are $ \mathcal{O}(n)$.   Now we prove the following lemma.

\begin{figure}[h!]
	
	\subfigure[Step 3]
	{
		\includegraphics[scale=0.35]{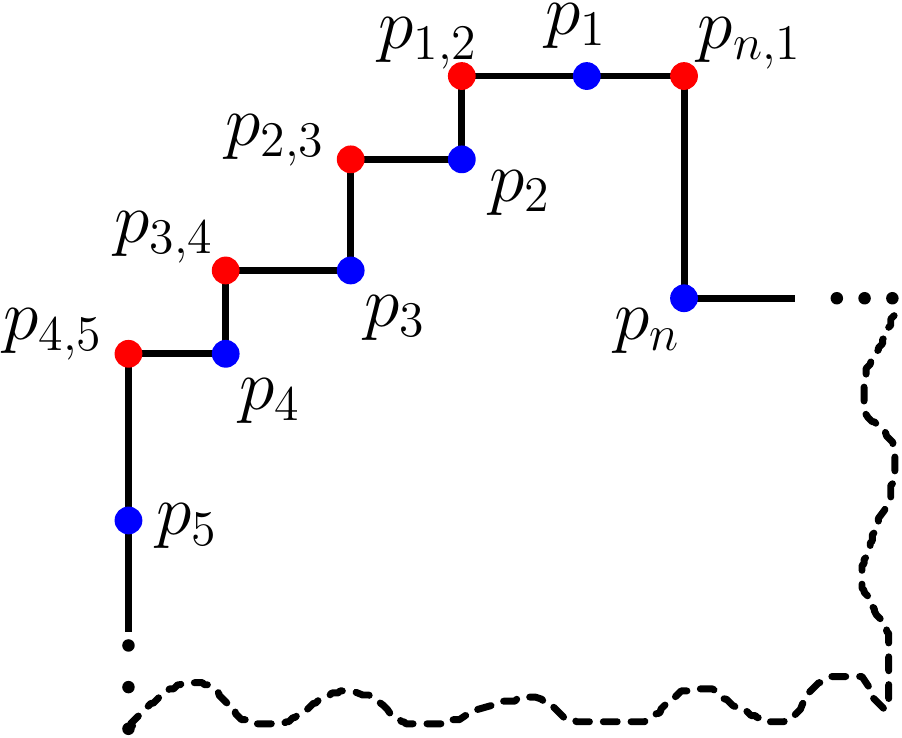}
		\label{fig-s3}
	}	
	\subfigure[Step 7-11]
	{
		\includegraphics[scale=0.4]{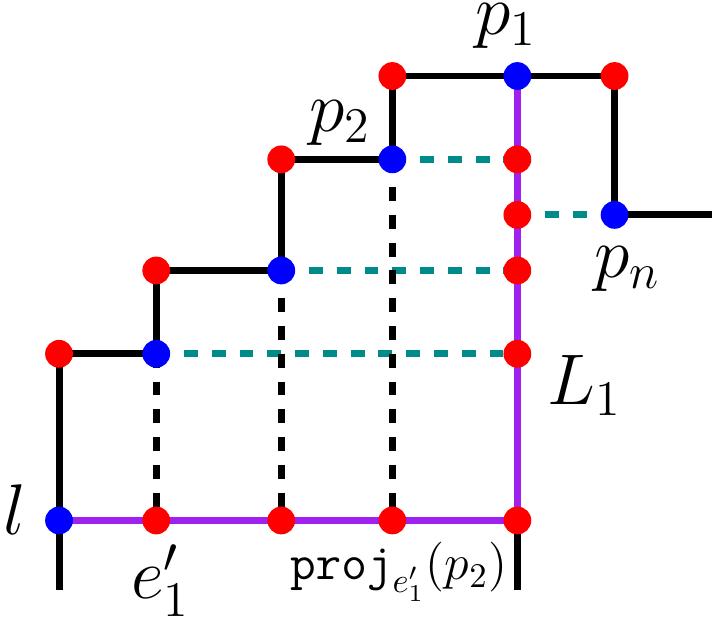}
		\label{fig-s8}
	}
	\subfigure[Step 12-15]
	{
		\includegraphics[scale=0.35]{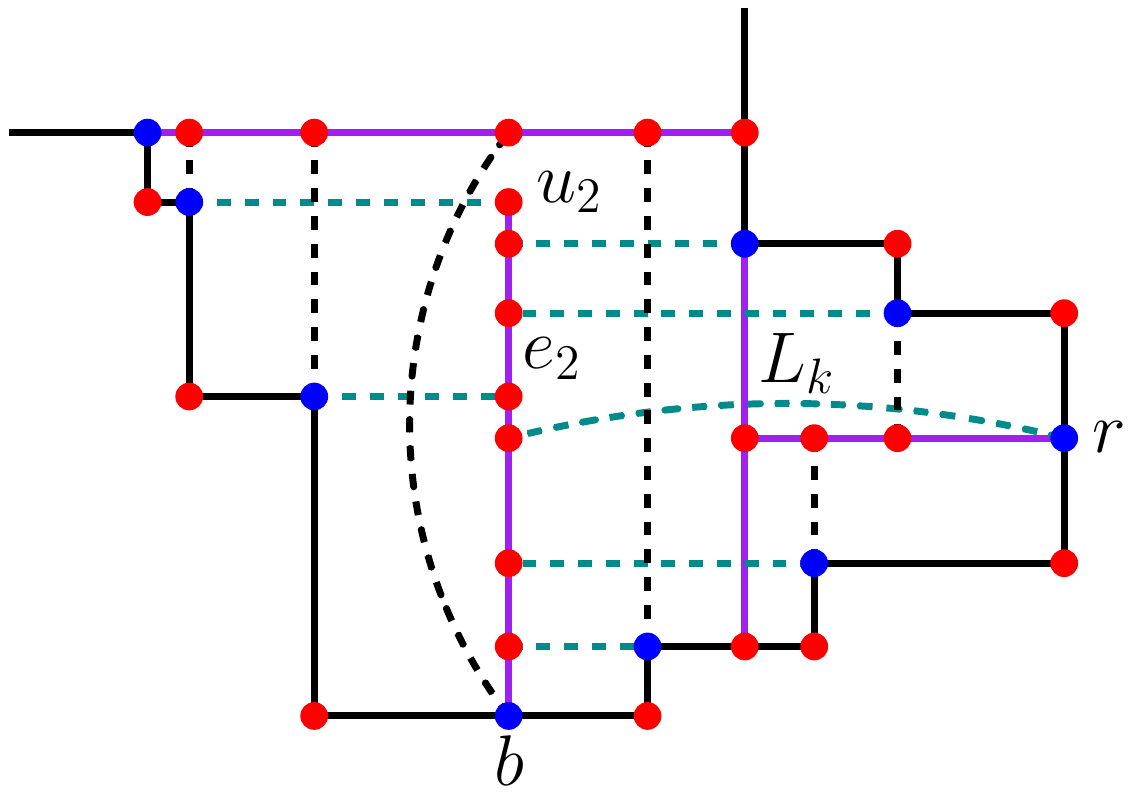}
		\label{fig-s15}
	}
	\subfigure[Step 16-19]
	{
		\includegraphics[scale=0.35]{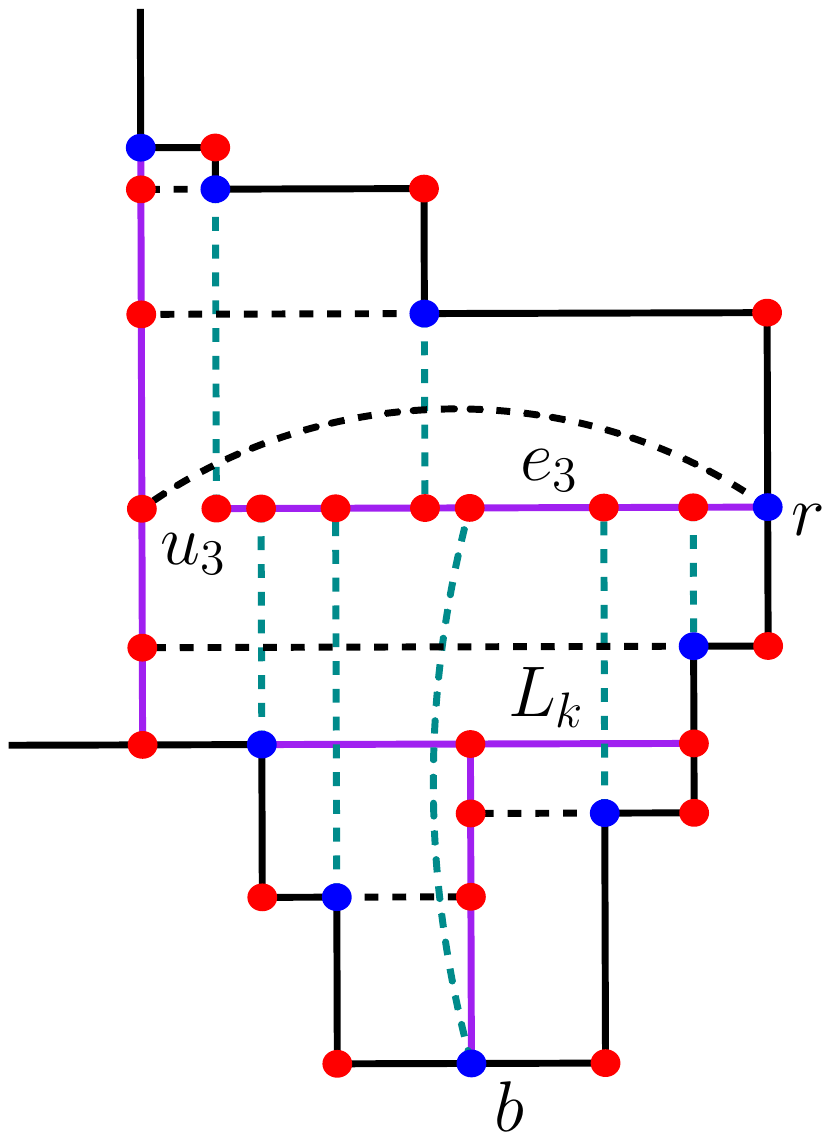}
		\label{fig-s19}
	}
	\subfigure[Step 20-23]
	{
		\includegraphics[scale=0.33]{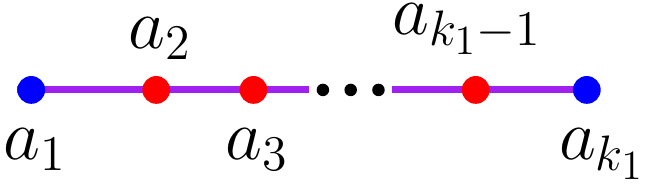}
		\label{fig-s14}
	}
	\subfigure[Step 29]
	{
		\includegraphics[scale=0.36]{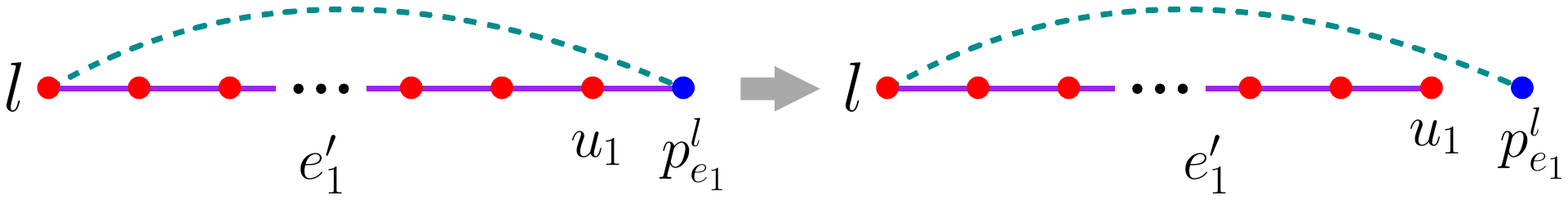}
		\label{fig-s191}
	}   
	\subfigure[Step 24-27]
	{
		\includegraphics[scale=0.4]{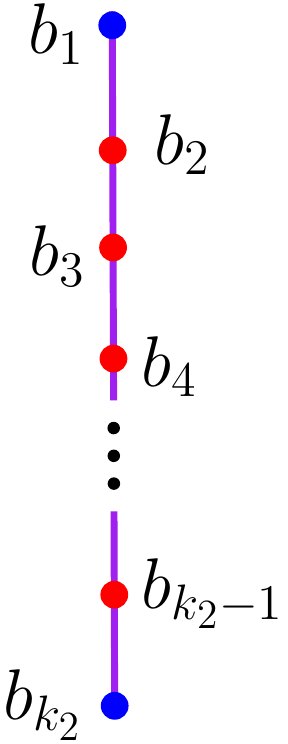}
		\label{fig-s16}
	}    
	\subfigure[Step 30]
	{
		\includegraphics[scale=0.4]{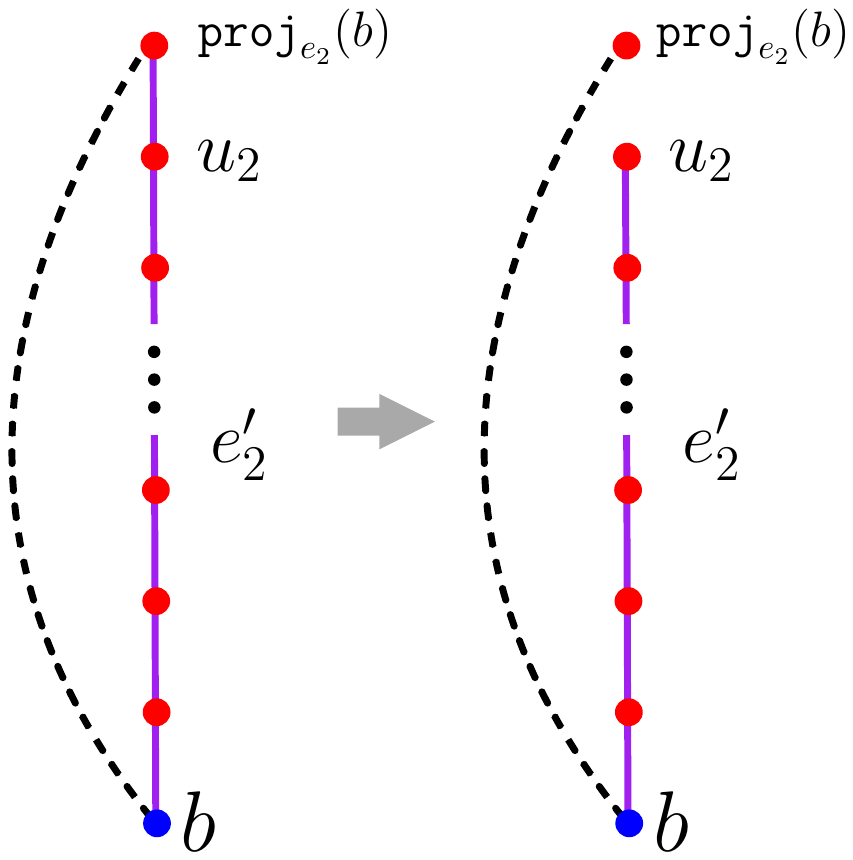}
		\label{fig-s193}
	}
	\subfigure[Step 31]
	{
		\includegraphics[scale=0.4]{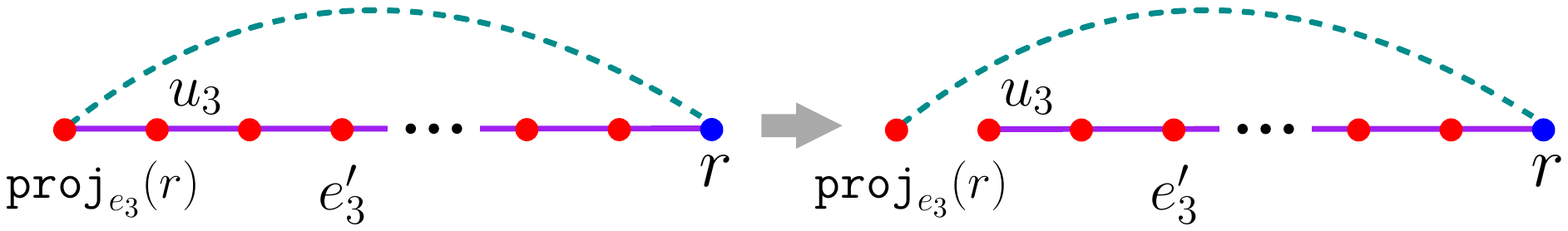}
		\label{fig-s192}
	}

	\caption[[Optional caption for list of figures]{Illustration of the Steps in \cref{algo_2}. We maintain following convention of colors. We use purple color while drawing the line segment of	 the set $\mathcal{L'} $. We use dashed black and dashed cyan line to denote vertical and horizontal projections, respectively of the points $ S $ to lines of $\mathcal{L'} $. Blue and red color points identify points from $ S $ and Steiner points, respectively.}
	\label{fig_step}
\end{figure}

\begin{lemma}\label{lem-S}
	For the point set $ S $, $G $  can be constructed in $ \mathcal{O}(n) $ time. 
\end{lemma}

\begin{proof}
	
	The construction of $ G $ from $ S $ consists of three Steps. In Step 1, we construct $ \mathcal{OCP(S)} $ from $ S $. As for each point $ p \in S $, we add exactly one Steiner point and draw two edges, $ \mathcal{OCP}(S) $ consists of $ 2n $ points including $ S $. So, Step 1 takes  $ \mathcal{O}(n) $  time. In Step 2, we construct a histogram partition $\mathcal{H}(\mathcal{OCP}(S))$ of $\mathcal{OCP}(S)$. By \cref{lem-L}, it needs $ \mathcal{O}(n) $ time. In the final Step, we apply \cref{algo_2} in $\mathcal{H}(\mathcal{OCP}(S))$ to construct our desired graph $G=(V,E)=(S \cup T, E)$.  Now we show \cref{algo_2} runs in $ \mathcal{O}(n) $ time. In this algorithm, Steps 1-4 take linear time. In Steps 7-11, for each point $ q\in S $, we perform orthogonal projections at most two times, i.e., we add at most two Steiner vertices and two edges. The points of $S$ are given in sorted order along their convex hull. Also, we have an ordered set of $ k $ line segments $ L_{1}, L_{2}, \dots, L_{k} $ with the ordering based on the construction of $\mathcal{H}(\mathcal{OCP}(S))$. Now, for any pair of points $ p_i $ and $ p_{i+1} $, where $ 1 \leq i \leq n $ if the point $ p_i $ has an orthogonal projection on $ L_m $ for some $ m $ then $ p_{i+1} $ can not have an orthogonal projection onto any line segment in $\mathcal{L} \setminus \{ L_{m-1}, L_{m}, L_{m+1} \} $. So it takes   $ \mathcal{O}(n+k)$ time to perform all the projections in Steps 7-11 by walking around the boundary of $ \mathcal{OCP}(S)$ once.  The Steps 12-15 occur only when both $H^r_{L_k} $ and $ e_2 $ exist. Now we have to do one more projection for each point of $ S \cap  H^r_{L_k}$ to $ e_2 $. So Steps 12-15 take linear time. Similarly, Steps 16-19 take linear time.  In Steps 20-25, we add edges to $ E $ by looking at each line segment of $ \{ L_{1}, L_{2}, \dots, L_{k}, e_1, e_2, e_3 \} $. As the number of projections is linear so the number of edges we add in Steps 20-25 is also linear. In Step 26, we delete some edges from $\{e_1, e_2, e_3, L_k\}$. So the total time complexity is $ \mathcal{O}(n+k)$.  As $ k \leq n $,  \cref{algo_2} produces $G$ in  $ \mathcal{O}(n) $ time. Hence the proof.  
\end{proof}

\subsection{$G$ is a $\mathsf{Manhattan~Network}$ }\label{sec-path}

To show that $G$ is a \pg, we have to prove that $G$ contains a shortest $L_1$ path between every pair of points in $ S $. Recall that $p(H)$ denotes the histogram containing $p \in S$ in $\mathcal{H}(\mathcal{OCP}(S))$ and $\mathcal{L}=\{ L_1, L_2, \dots L_k\} $ denotes the set of $ k $ segments inserted in $ \mathcal{OCP}(S) $ while constructing $\mathcal{H}(\mathcal{OCP}(S))$. First we prove the following lemma.

\begin{lemma}\label{lem-cut}
	For any two points $w$ and $z$ in $S$, if $w(H)\neq z(H)$ then there always exist lines $L$ and $L'$ such that $(i)~\mathtt{proj}_{L}(w) \in L,~\mathtt{proj}_{L'}(z) \in L'$ and $(ii)$ if we draw  a line $L^*$ that contains line $L$ (respectively, $L'$ ) then $w$ and $z$ belong to opposite sides of $L^*$.
\end{lemma}

\begin{proof}
	Let $w$ and $z$ be two points in $S$ such that $w(H)\neq z(H)$. Without loss of generality we assume that $x(w) < x(z)$. By our construction of $\mathcal{H}(\mathcal{OCP}(S))$, $x(L_1)<x(L_3)< \dots$ and $y(L_2)>y(L_4)> \dots$. If $x(w) \leq x(L_1)$ then $L=L_1$. Let $x(w) \geq x(L_1)$ and $i$ be the largest integer such that $x(L_{i})\leq  x(w)$. If $L_{i+2}$ exists and $\mathtt{proj}_{L_{i+2}}(w)  \in L_{i+2}$ then $L= L_{i+2}$, else $L= L_{i+1}$. Similarly let $j$ be the largest integer such that $x(L_{j})\leq  x(z)$. If $\mathtt{proj}_{L_{j}}(z) \in L_{j}$ then $L'= L_{j}$, else $L'= L_{j+1}$. 
\end{proof}

Now we prove the following lemma.

\begin{lemma}\label{lem-path}
	For each pair of points $  p_i$ and $ p_j $ of $ S $ where $ 1 \leq i, j 
	\leq n $, there exists a shortest $ L_1 $ path in $ G $ between them.
\end{lemma}

\begin{proof}
	$\mathcal{OCP}(S) $ consists of four $ xy$-monotone chains $ C_{rt}, C_{tl}, C_{lb}, \text{and}~ C_{br} $. Let $
	p_{i} $ and $ p_{j}$ be two arbitrary points of $ S $ where $ 1 \leq i, j \leq n $.  Let $ \pi_G (a, b)=\langle a, \ldots, v_{i},\ldots,b \rangle$ denotes a shortest $L_1$ path between a pair of vertices $ a $ and $ b $ in $ G $. Let $P_1$ and $P_2$ be two paths from $a$ to $b$ and $b$ to $c$, respectively. By $P_1 \leadsto P_2$ we mean the path from $a$ to $c$ that is   obtained by concatenating the paths $P_1$ and $P _2$.  The proof of this theorem can be divided into Case A and Case B. 			
	\begin{description}
		\item[Case A: Both $ \bm{p_{i}} $ and $ \bm{p_{j}}$ belong to the same $ xy$-monotone chain:] Each $xy$-monotone chain of the ortho-convex polygon $\mathcal{OCP}(S) $ is a Manhattan network for the points it contains.

		
		\item[Case B: $ \bm{p_{i}} $ and $ \bm{p_{j}}$  belong to different chains:]  
		
		We divide this case into two subcases B.1. and B.2.

		\item[Case B.1. $ \bm{p_{i}(H)}= \bm{p_{j}(H)}$, i.e., $ \bm{p_{i}} $, $ \bm{p_{j}}$ belong to the same histogram]
		{~

			\begin{description}

				\item[\textbf{(1)}  $ \bm{{p_{i}}, {p_{j}} \in l(H)}$:]
				
				If $ p_i \in  C_{tl}  $, $ p_j \in C_{lb} $ then $\pi_G (p_i,p_j)= \langle p_i,\mathtt{proj}_{e'_1}(p_i)   \rangle  \\ \leadsto \pi_G (\mathtt{proj}_{e'_1}(p_i) ,\mathtt{proj}_{e'_1}(p_j) )  \leadsto \langle \mathtt{proj}_{e'_1}(p_j) , p_j \rangle $.

				\item[\textbf{(2)} $ \bm{{p_{i}}, {p_{j}} \in {r(H)}}$:] For Types 1, 2, or 3, if $ p_i \in  C_{rt}  $~and $ p_j \in C_{br} \cup C_{lb} $, then  $\pi_G (p_i,p_j)=\langle p_i,\mathtt{proj}_{e'_3}(p_i)  \rangle  \leadsto \pi_G (\mathtt{proj}_{e'_3}(p_i) ,\mathtt{proj}_{e'_3}(p_j))   \leadsto \langle \mathtt{proj}_{e'_3}(p_j), p_j \rangle $. 
				For Type 3,  if $ p_i \in  C_{lb}  $ and $ p_j \in C_{br} $ then $\pi_G (p_i,p_j)= \langle p_i,\mathtt{proj}_{L_k}(b) \rangle  \leadsto \pi_G (\mathtt{proj}_{L_k}(b),p_j) $.
				For Type 4, if $ p_i \in  C_{lb}  $ and $ p_j \in C_{br} $ then $\pi_G (p_i,p_j)= \langle p_i,\mathtt{proj}_{e'_2}(p_i) \rangle  \leadsto \pi_G (\mathtt{proj}_{e'_2}(p_i),\mathtt{proj}_{e'_2}(p_j)) \\ \leadsto \langle \mathtt{proj}_{e'_2}(p_j), p_j \rangle $.

				\item[\textbf{(3)} $ \bm{{p_{i}}, {p_{j}} \in {b(H)}}$:] For Types 1 or 3, if $ p_i \in  C_{lb}  $~and $ p_j \in C_{br} \cup C_{rt} $, then $\pi_G (p_i,p_j)=\langle p_i,\mathtt{proj}_{e'_2}(p_i) \rangle  \leadsto \pi_G (\mathtt{proj}_{e'_2}(p_i),\mathtt{proj}_{e'_2}(p_j))  \leadsto \langle \mathtt{proj}_{e'_2}(p_j), p_j \rangle $. 
				For Type 1, (i) if $ p_i \in  C_{rt}  $ and $ p_j \in C_{br} $ then $\pi_G (p_i,p_j)= \langle p_i,\mathtt{proj}_{L_k}(r) \rangle  \leadsto \pi_G (\mathtt{proj}_{L_k}(r),p_j) $ or $\pi_G (p_j,p_i)= \langle p_j,\mathtt{proj}_{L_{k-1}}(p_j) \rangle 
				\leadsto \pi_G (\mathtt{proj}_{L_{k-1}}(p_j),p_i) $.
				(ii) if $ p_i \in  C_{lb}  $ and $ p_j \in C_{rt} $ then the shortest $ L_{1} $ path between $ \bm{p_{i}} $ and $ \bm{p_{j}}$ in $ G $ is  $\pi_G (p_i,p_j)=\langle p_i,\mathtt{proj}_{e'_2}(p_i) \rangle  \leadsto \pi_G (\mathtt{proj}_{e'_2}(p_i),\mathtt{proj}_{e'_2}(p_j))  \leadsto \langle \mathtt{proj}_{e'_2}(p_j), p_j \rangle $ or $\pi_G (p_j,p_i)= \langle p_i,\mathtt{proj}_{L_{k-1}}(p_i) \rangle  \leadsto \pi_G (\mathtt{proj}_{L_{k-1}}(p_i),p_j) $.
				For Types 2 or 4 as $ b(H)=r(H) $, it is similar as subcase (2) of B.1.

				\item[\textbf{(4)} $ \bm{{p_{i}}, {p_{j}} \notin \{l(H), b(H), r(H)\}}$:] Let these histograms  contain two elements say $ L $ and $ L' $ of $ \mathcal{L} $. In this case $\pi_G (p_i,p_j)= \langle p_i,\mathtt{proj}_{L}(p_i) \rangle \leadsto \pi_G (\mathtt{proj}_{L}(p_i),\mathtt{proj}_{L}(p_j))
				\leadsto \langle \mathtt{proj}_{L}(p_j),p_j \rangle $ or  $\pi_G (p_i,p_j)= \\ \langle p_i,\mathtt{proj}_{L'}(p_i) \rangle \leadsto \pi_G (\mathtt{proj}_{L'}(p_i),
				\mathtt{proj}_{L'}(p_j))  \leadsto \langle \mathtt{proj}_{L'}(p_j),p_j \rangle $.
				
			\end{description}

		}

		\item[Case B.2. $ \bm{p_{i}(H)} \neq \bm{p_{j}(H)}$:]
		
		First, we find line segments $ L, L' \in \mathcal{L} $ such that $(i)$ $ L $ can see $ p_i $, $ L'$ can see $ p_j $, and $(ii)$  if we draw  a line $L^*$ that contains line $L$ (respectively, $L'$ ) then $p_i$ and $p_j$ belong to opposite sides of $L^*$. By \cref{lem-cut} both $ L$ and $L' $ exist in $ \mathcal{L} $ but it may happen that $ L=L' $ e.g., for the points $ p_2 $ and $ p_n, $ $ p_2(H) \neq p_n(H) $ with $ L=L'$. By the construction of $ G $, both $ \mathtt{proj}_{L}(p_i) $ and $ \mathtt{proj}_{L'}(p_j) $ belong to $ T \subset V$.
		We complete this case by proving following lemma. 
		\begin{lemma}
			Let $ w $ and $ z $ be two points in $ S $ such that $ {w(H)} \neq {z(H)}$. Also let $ L $ and $ L' $ be two segments such that $(i)~\mathtt{proj}_{L}(w) \in L,~\mathtt{proj}_{L'}(z) \in L'$ and $(ii)$ if we draw  a line $L^*$ that contains line $L$ (respectively, $L'$ ) then $w$ and $z$ belong to opposite sides of $L^*$.  Then there exist a shortest $ L_1 $ path between $ \mathtt{proj}_{L}(w)$ and $ \mathtt{proj}_{L'}(z) $ in $ G $.
		\end{lemma}

			\noindent{ \emph{Proof.}} Without loss of generality, we assume that $ x(w)<x(z) $. If $ L=L' $ then $\pi_G (\mathtt{proj}_{L}(w),\mathtt{proj}_{L'}(z))$ is along the line $ L $.
			For example, if we take $ w=p_2 $ and $ z=p_n $ then $ L=L'=L_1 $. So we are left with the case when 
			$ L\neq L' .$  For example, in \cref{fig-graph}, considering $ l $ as $ w $ and $ r $ as $ z $ we find $ L=L_1$ and $ L'=L_k $.  Rest of the proof can be divided into two cases. Recall that $ \{L_1, L_2, \dots L_k\} $ are the segments inserted in $ \mathcal{OCP}(S) $ while constructing $\mathcal{H}(\mathcal{OCP}(S))$. The point set $ \{q_i \colon 1 \leq i \leq k\} $ comes from the construction of $\mathcal{H}(\mathcal{OCP}(S))$. Assuming $ l=q_0 $, $ L_i $ is the segment with end points $ q_{i-1} $ and $ q_{i} $, where $ 1 \leq i \leq k $.
			\\
			\underline{\textbf{Case 1. $ \bm{L}$ is vertical$\bm{\colon}$}} 	  	
			\noindent Let $ L=L_m $ for some $ m, 1 \leq m \leq k $. So $ L_{m}= \overline{q_{m-1}q_{m}} $. By the construction of $\mathcal{H}(\mathcal{OCP}(S))$, $ q_m $ is not only a point on the boundary of $ \mathcal{OCP}(S) $ but also there exists a point say $ p_j $ in $ S $ such that $ q_m \in \overline{p_{j,j-1 } p_{j}} $. Now we divide this case into following two subcases.
			
			\begin{description}
				\item[Case 1.1 $ \bm{L'}$ is vertical:] By similar argument as $ L $, there exists a point $ p_{j'} $ such that $ y(p_{j'}) \geq y(z)$ and $ p_{j'} \in L' $. For this case, a shortest $ L_1 $ path between $ \mathtt{proj}_{L}(w)$ and $ \mathtt{proj}_{L'}(z) $ in $ G $ is $ \pi_G (\mathtt{proj}_{L}(w),q_{m}) \leadsto \pi_G (q_{m}, p_j)  \leadsto \pi_G ( p_j, p_{j'}) \leadsto \pi_G (p_{j'}, \mathtt{proj}_{L'}(z))  $. By repeatedly applying this argument we can find $ \pi_G ( p_j, p_{j'}) $. For an illustration, see \cref{fig-vv}.
				\noindent
				\begin{figure}[ht!]
					
					\subfigure[]
					{
						\includegraphics[scale=0.550]{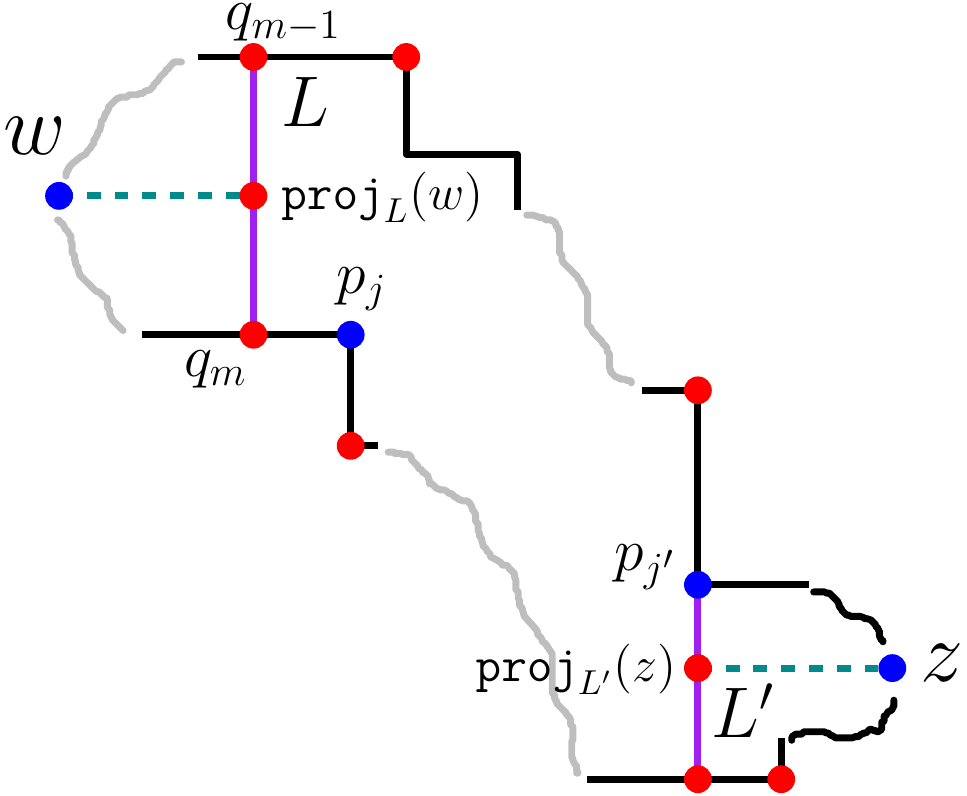}
						\label{fig-vv}
					}	
					\subfigure[]
					{
						\includegraphics[scale=0.550]{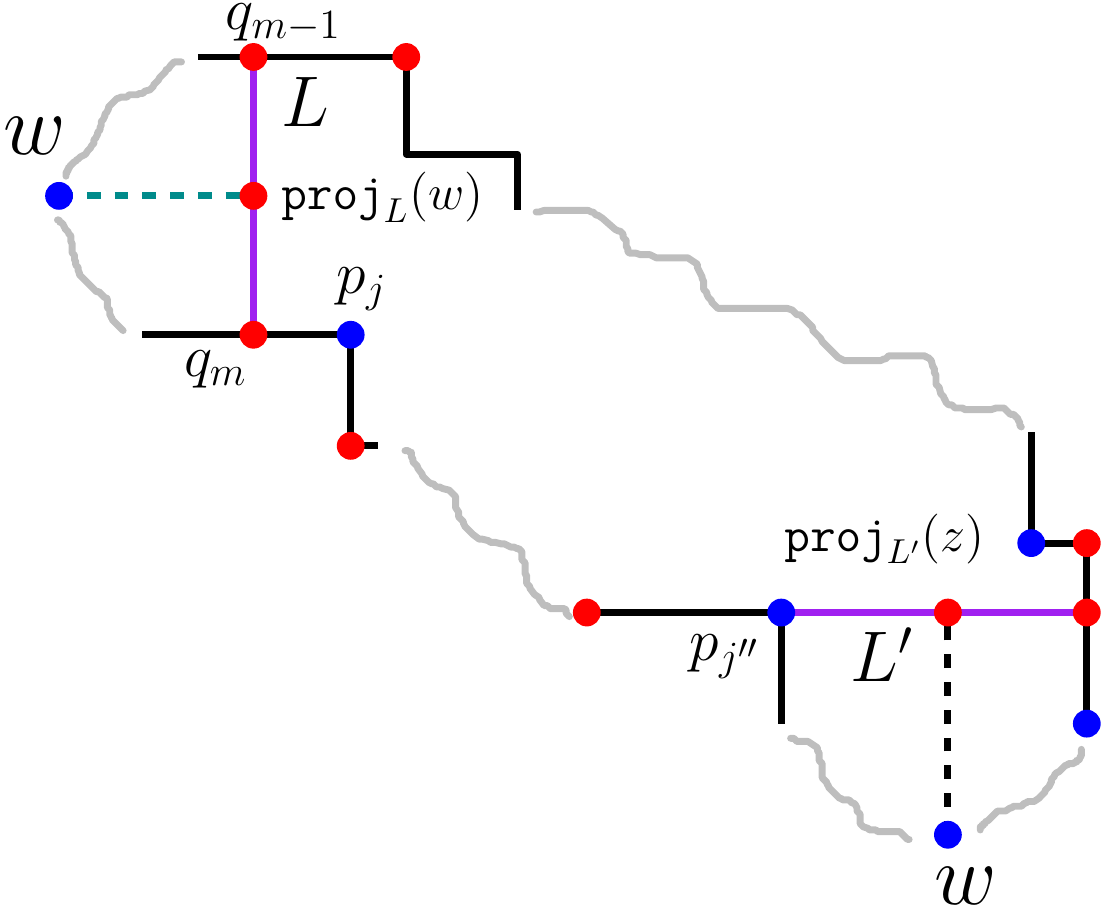}
						\label{fig-vh}
					}
					
					\caption{(a) Both $ L $ and $ L' $ are vertical. (b) $ L $ is vertical, $ L' $ is horizontal.} 
					\label{}
				\end{figure}

				\item[Case 1.2 $ \bm{L'}$ is horizontal:] By similar argument as $ L $, there exists a point $ p_{j''} $ such that $ x(p_{j''}) \geq x(w)$ and $ p_{j''} \in L' $. Rest of this case is similar as case 1.1. Here a shortest $ L_1 $ path between $ \mathtt{proj}_{L}(w)$ and $ \mathtt{proj}_{L'}(z) $ in $ G $ is $ \pi_G (\mathtt{proj}_{L}(w),q_{m}) \leadsto \pi_G (q_{m}, p_j)  \leadsto \pi_G ( p_j, p_{j''}) \leadsto \pi_G (p_{j''}, \mathtt{proj}_{L'}(z))  $. By repeatedly applying this argument we can find $ \pi_G ( p_j, p_{j''}) $.  For an illustration, see \cref{fig-vh}.
			\end{description}
			
			\noindent \underline{\textbf{Case 2. $ \bm{L}$ is horizontal $\bm{\colon}$}} 
			Proof for this case is similar as case 1.

	\end{description}
\end{proof}

\subsection{Planarity of $ G $}\label{sec-planar}
In this section, we show that the graph $ 	G=(V,E) $ is planar by providing a planar embedding. For an illustration, see \cref{fig-G'}.

\begin{figure}[ht!]
	
	\subfigure[]
	{
		\includegraphics[scale=0.315]{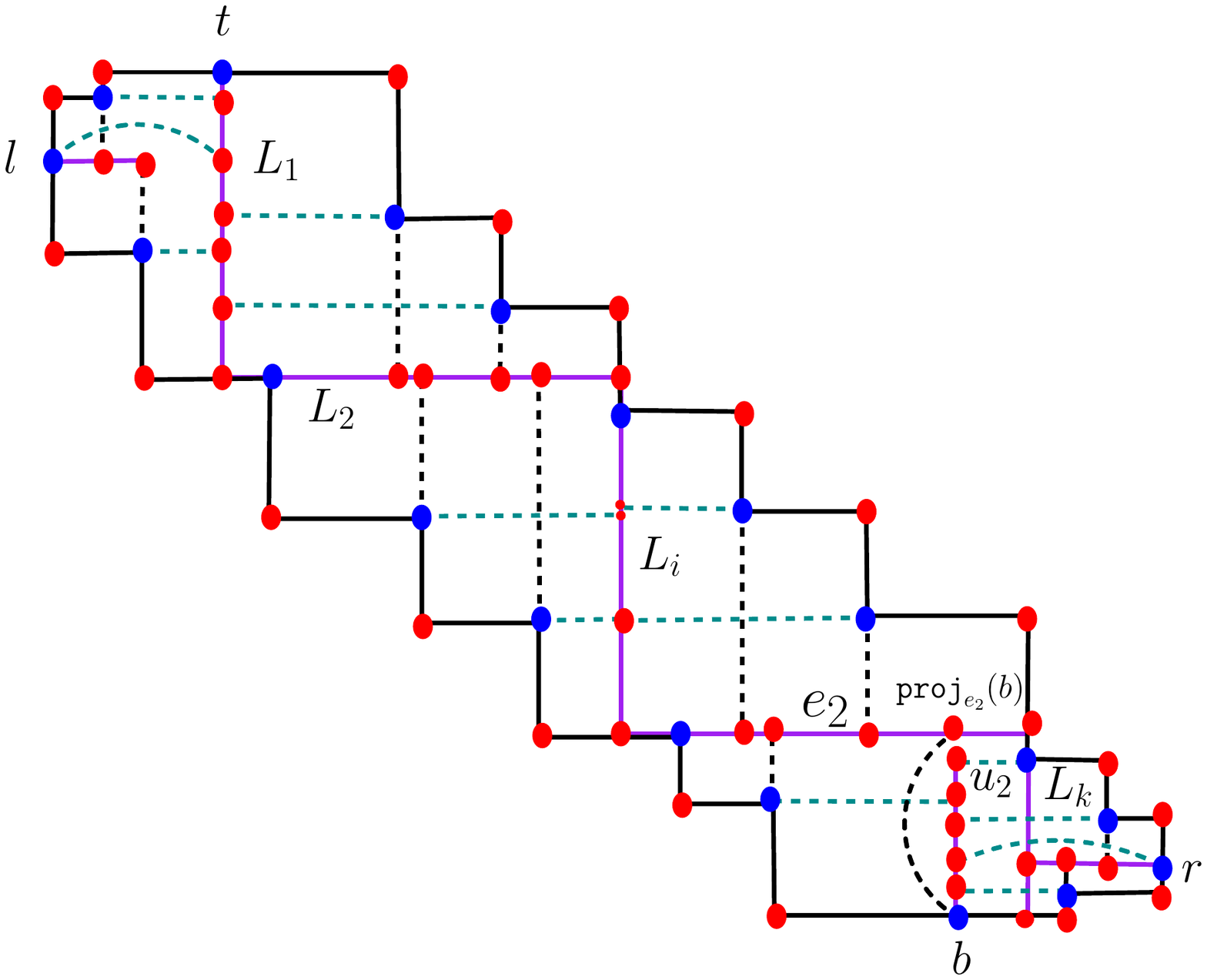}
		\label{fig-graph}
	}	
	\subfigure[]
	{
		\includegraphics[scale=0.315]{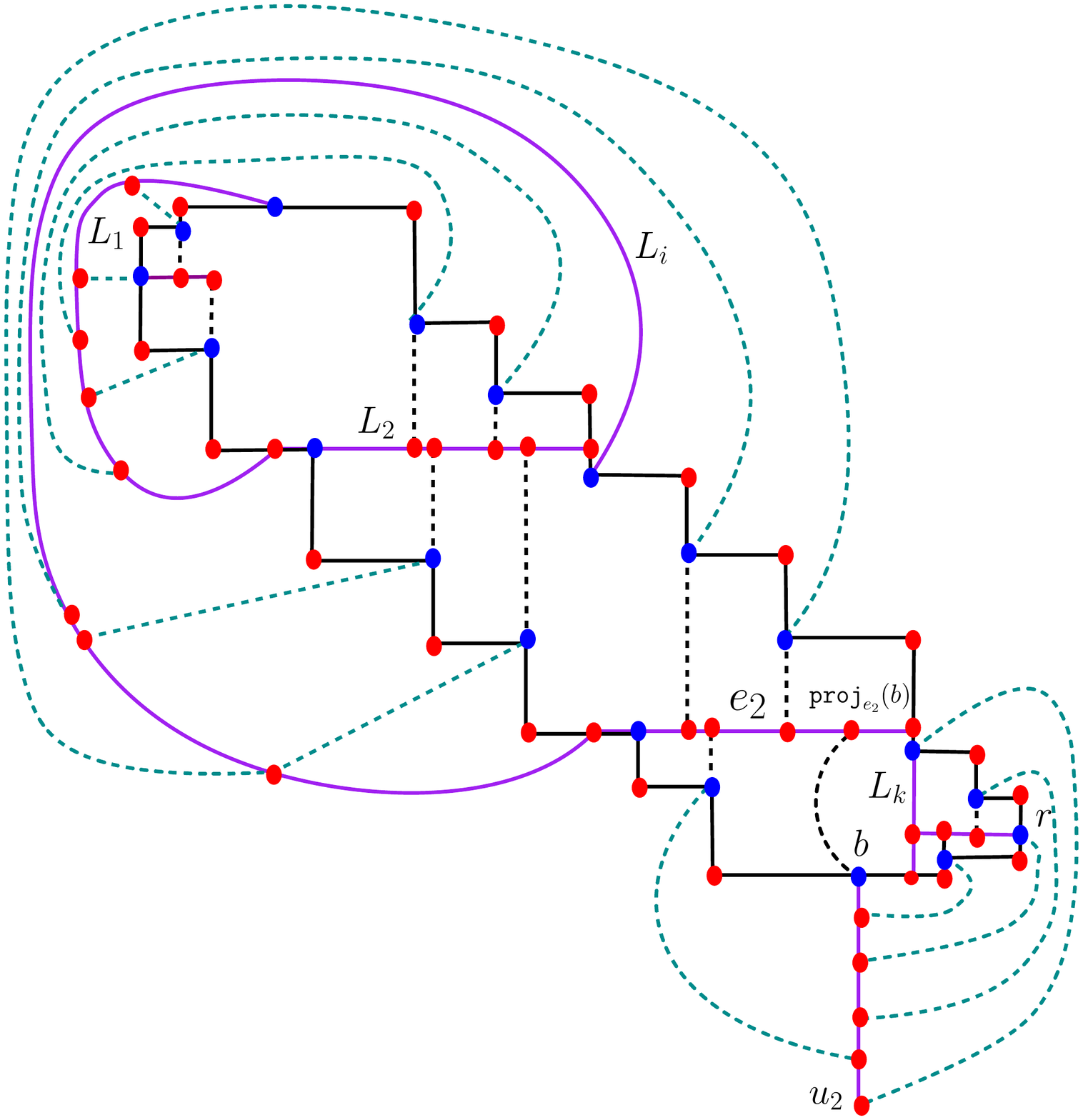}
		\label{fig-G'}
	}
	
	\caption{(a) Output $G$ of \cref{algo_2} for point set in blue color. (b) Planar embedding of $ G $.} 
	\label{}
\end{figure}

 \par
The {\em union} of two graphs $ G_{1}= (V_{1}, E_{1}) $ and $ G_{2}=(V_{2}, E_{2}) $  is defined as the graph $ (V_1 \cup V_2, E_1 \cup E_2)$ \cite{west1996introduction}. We will make use of the following theorem
regarding planar graphs.

\begin{theorem} \cite{gibbons1985algorithmic}\label{theo_planar}
	A planar embedding of a graph can be transformed into another planar embedding such that any specified face becomes the exterior face. 
	
\end{theorem}

\paragraph{\textbf{Relation to k-plane graphs} \cite{garcia2015geometric}} A geometric graph $G = (V, E)$ is said to be $k$-plane garph for some $k \in N$ if $ E $ can be partitioned into $k$ disjoint subsets,$E=E_1 \cupdot E_2 \cupdot \dots \cupdot E_k$, such that $G_1 = (V, E_1), G_2=(V, E_2), \dots, G_k = (V, E_k )$ are all plane graphs, where $\cupdot$ represents the disjoint union. For a finite general point set $ P $ in the plane, $\mathcal{G}_k(P)$ denotes the family of $k$-plane graphs with vertex set $P$. As per as our constrcution, the graph we construct to form a \pgt~for convex point set is basically a 2-plane graph.

\begin{theorem}\label{theo-G}
	Graph $ G $ computed in \cref{algo_2} is planar.
\end{theorem}
\begin{proof}
	We decompose  $ G $ into two subgraphs $ H $ and $ K $ such that $ G= H \cup  K $.
	This decomposition depends on the line $ L_{k-1} $. In order to construct the histogram partition in $ \mathcal{OCP}(S) $, $ L_{k-1} $ may be horizontal or vertical. For Types 1 and 2, $ L_{k-1} $ is horizontal. For Types 3 and 4, $ L_{k-1} $ is vertical. We  analyse each of the following two cases.
	
	\begin{figure}[ht!]
		\hspace{5mm}
		\subfigure[]
		{
			\includegraphics[scale=0.31]{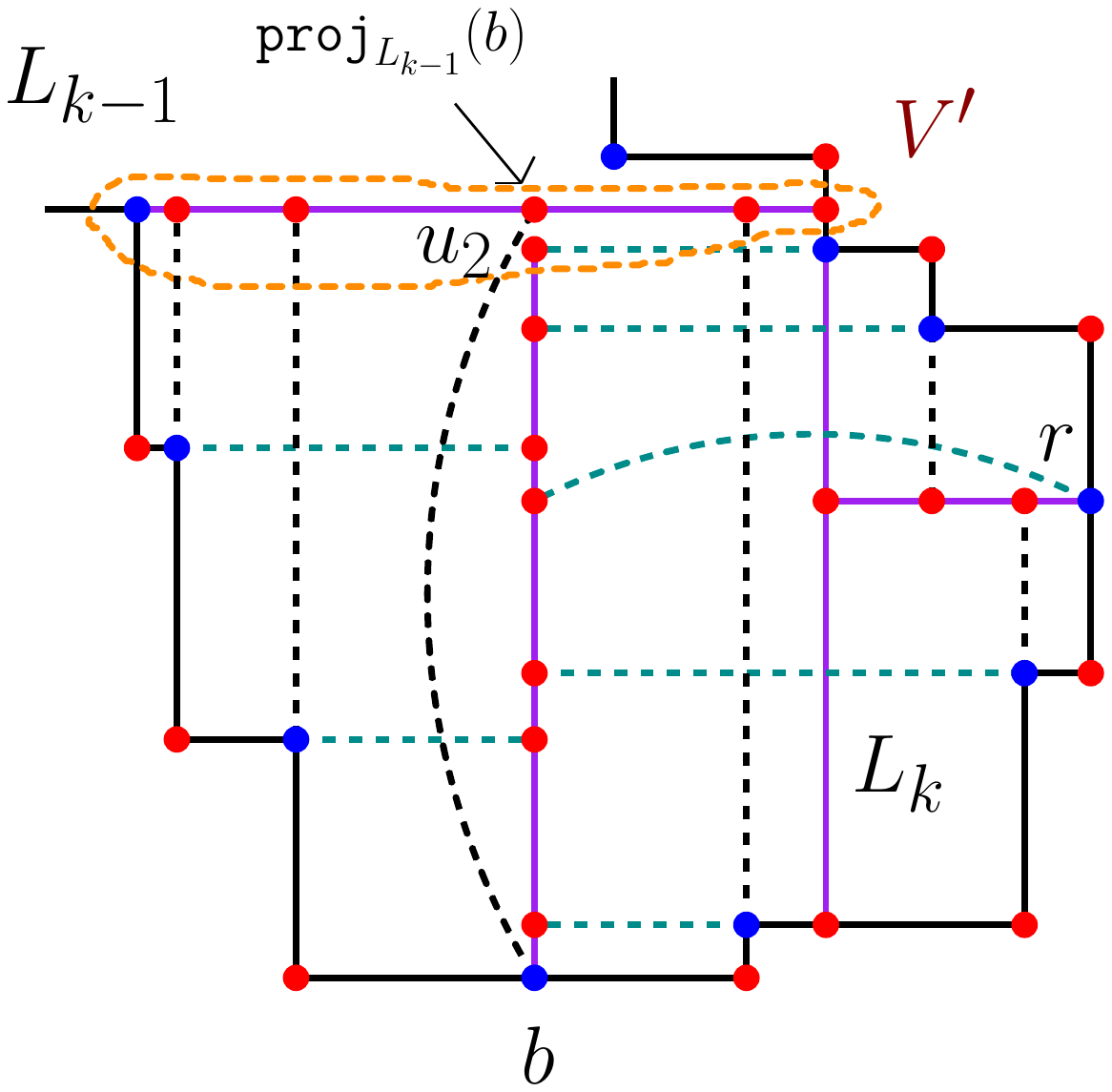}
			\label{}
		}
		\hspace{15mm}	
		\subfigure[]
		{
			\includegraphics[scale=0.31]{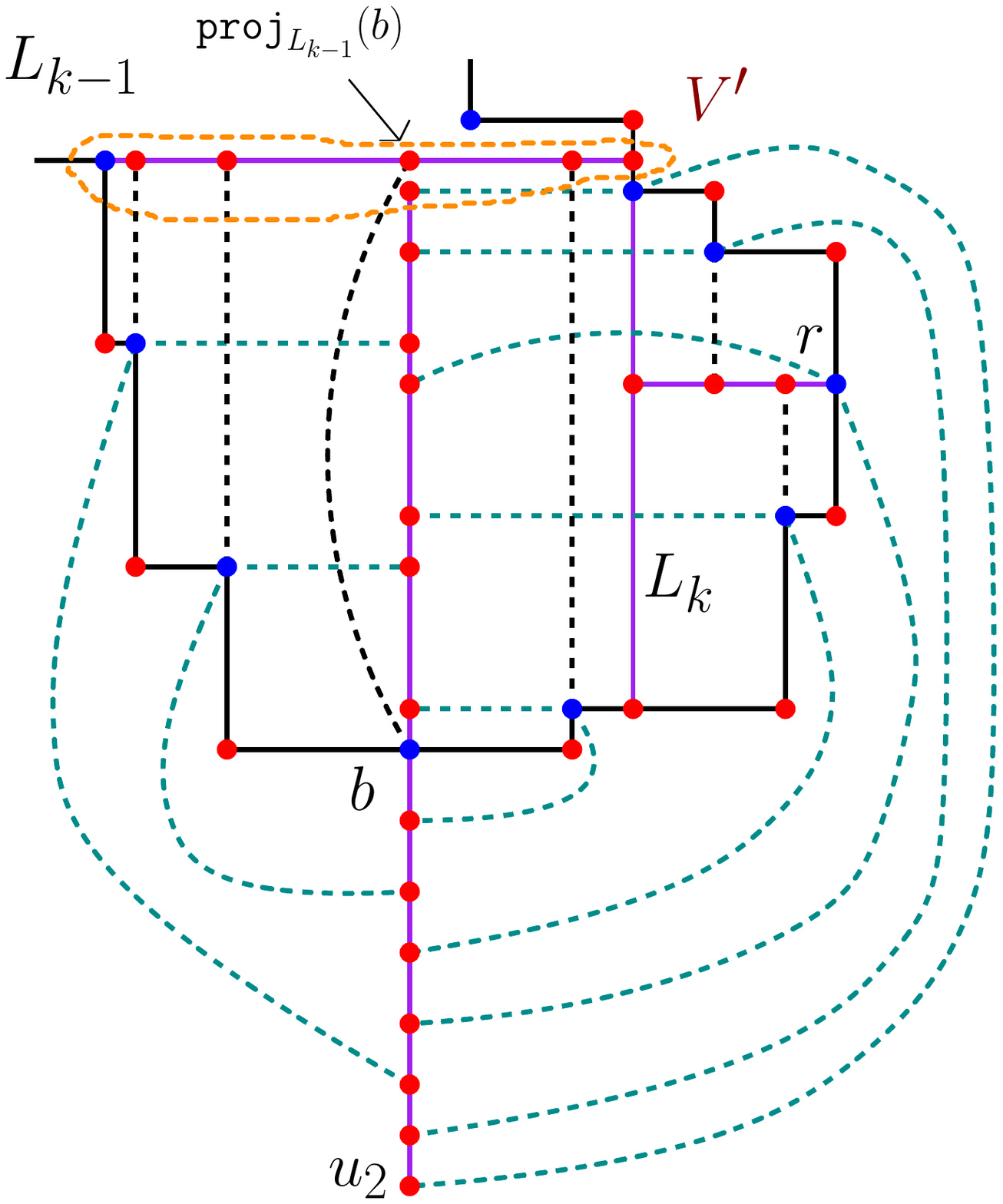}
			\label{}
		}
		
		\caption{ (a) The subgraph $ \bm{K} $ of $ G $ for \textbf{Type 1 } with the exterior face containing $ V'$. (b) planar embedding of $ \bm{K} $. The edges of $E_b$ are shown by  dashed cyan segment.}
		\label{fig-k1}
	\end{figure}	 
	
	\noindent \underline{\textbf{Case 1. $ \bm{L_{k-1} }$ is horizontal$\bm{\colon}$}} 		
	\noindent	$ H $ and $ K $ are the subgraphs of $ G $ induced by the vertices lying above and below, respectively of the line segment $ L_{k-1} $, i.e., $V(H) =\{v \colon v \in V,~y(v) \geq y(L_{k-1})\}$, where $y(L_{k-1})$ is the $ y $-coordinate of any point on the segment $ L_{k-1}$. Similarly, $V(K)= \{v \colon v \in V,~y(v) \leq y(L_{k-1})\}$. Let $ V'= V(H) \cap V(K)$. We want to show that $ G $ is planar, i.e., there exists a planar embedding $ G' $ of $ G $. If we are able to show that there exist two planar embeddings, $ H' $  for $ H $ and  $ K' $  for $ K $, such that  $ V' $ belongs to the exterior faces of both $ H' $ and $K'$, then we can obtain a  planar embedding $ G' $ of $ G $ by attaching the embeddings of $ H' $ and $ K' $ along the exterior face. Now our target is to show that $ H $ and $K $ have planar embeddings $ H'$ and $ K' ,$ respectively, such that $ V' $ is contained in the exterior face of   both $ H' $ and $ K' $. We define  $ V^f \subseteq V$ to denote the set of vertices of $ G $ along the boundary of $ \mathcal{OCP}(S) $. To Get a planar embedding of $ G $, we prove following two lemmas.
	
	\begin{figure}[ht!]
		\hspace{5mm}
		\subfigure[]
		{
			\includegraphics[scale=0.31]{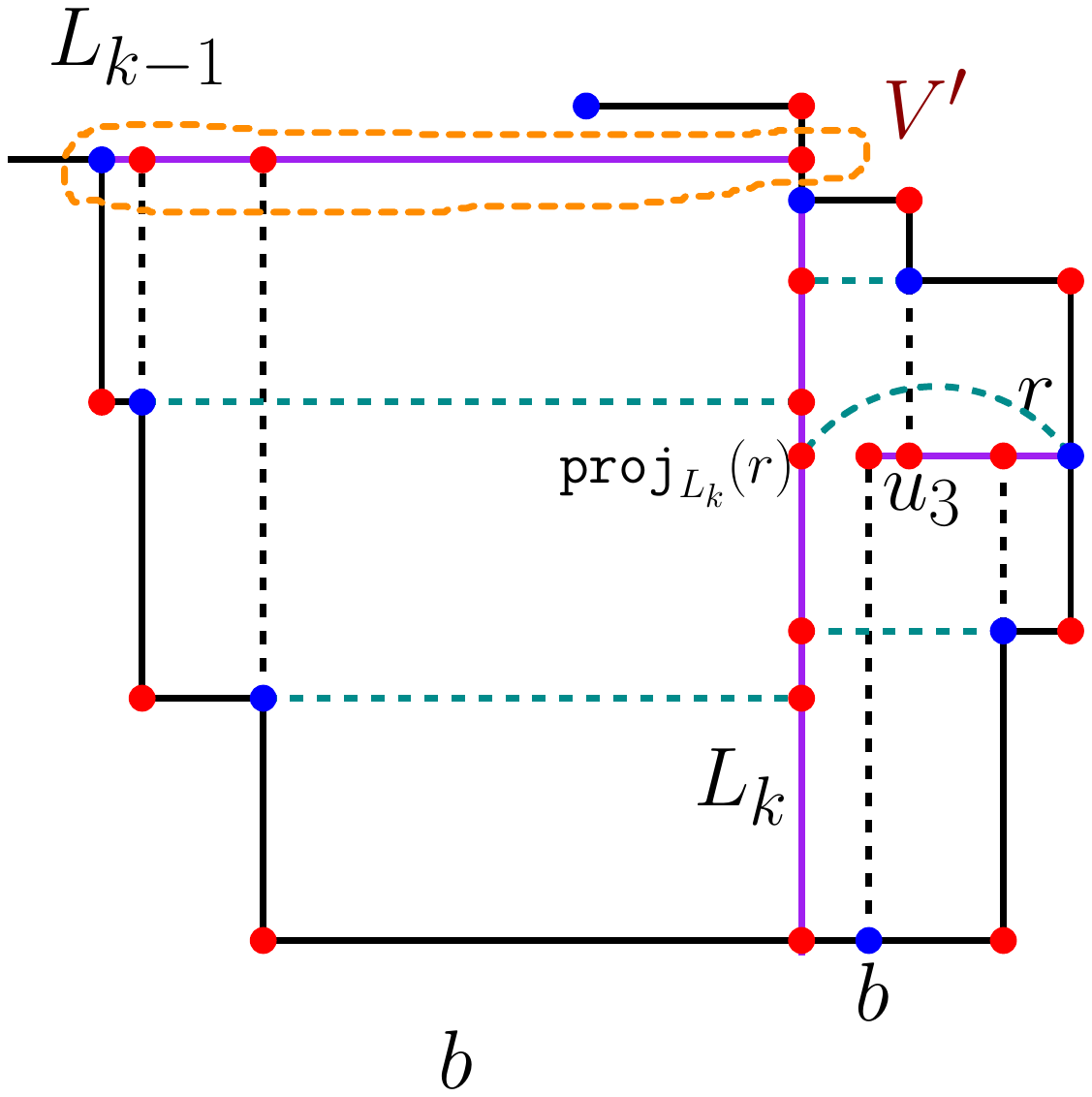}
			\label{}
		}
		\hspace{15mm}	
		\subfigure[]
		{
			\includegraphics[scale=0.31]{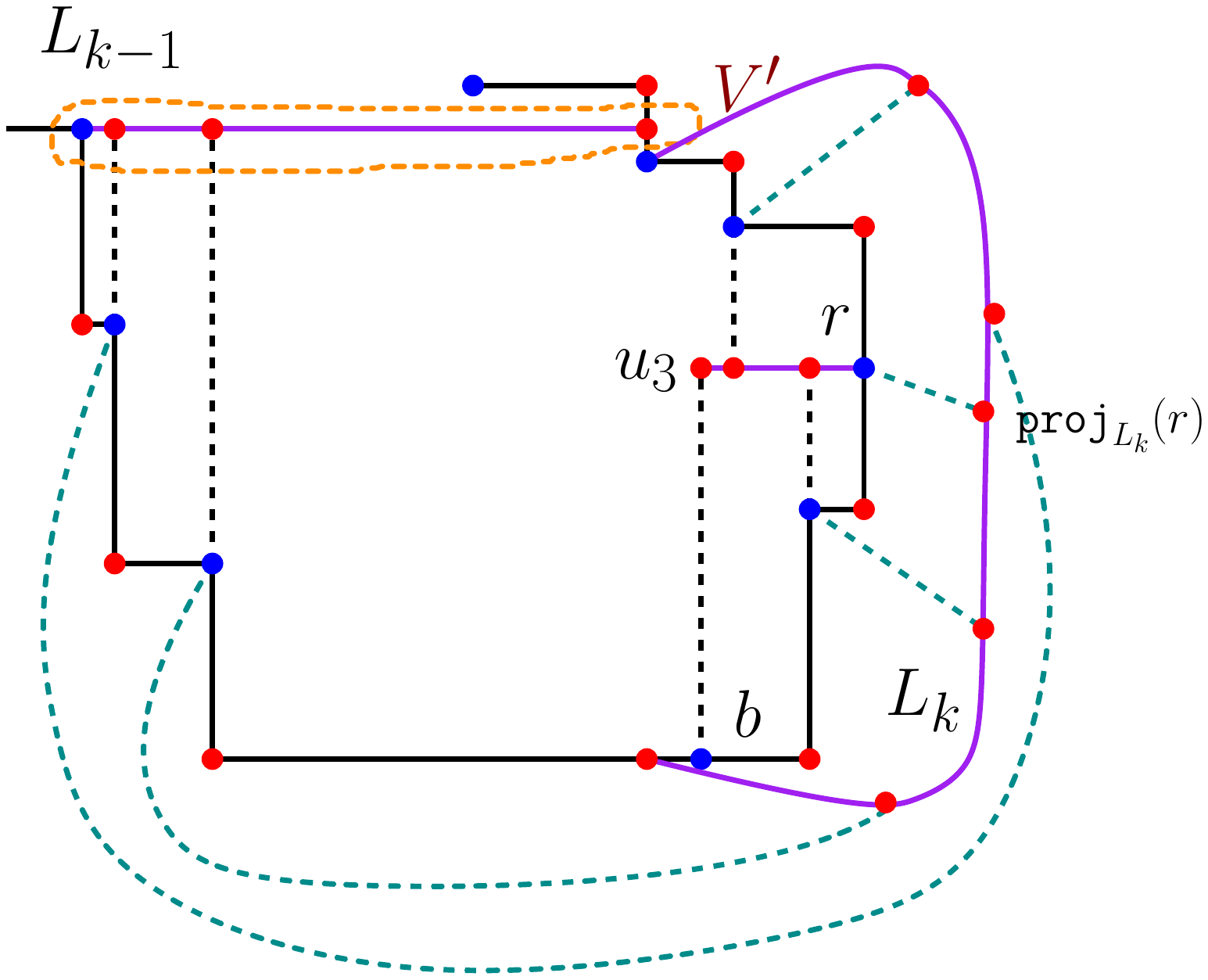}
			\label{}
		}
		
		\caption{(a) The subgraph $ \bm{K}  $ of $ G $ for \textbf{Type 2} with the exterior face containing $ V' $. (b) planar embedding of $ \bm{K} $. The edges of $E_k$ are shown by dashed cyan segment.}
		\label{fig-k2}
	\end{figure}

	\begin{lemma}
		$ K $ has a planar embedding $ K' $ such that $ V' $ is contained in the exterior face of   $ K'$.
	\end{lemma}
	
	\begin{proof} Let $ V^f_k \subseteq V$ be the set of vertices in $ G $ along the exterior face of $ K $. So $ V^f_k = (V^f \cap V(K)) \cup V'$. For Type 1, let $ E_b $ be the set of horizontal edges that have at least one adjacent vertex on the segment $ e'_2 =\overline{b~\mathtt{proj}_{L_{k-1}}(b)}$. In this case, we draw the edges $ E_b $ in the exterior face of $ K $ in such a way that we obtain a planar embedding of $ K $.  In the planar embedding, all Steiner points on the line segment $\overline{bu_2}$ will go to the exterior of the polygon along with its adjacent edges. For Type 2, let $ E_k $ be the set of horizontal edges that have at least one adjacent vertex on the line $ L_k $. In this case, we draw the edges $ E_k $ in the exterior faces of $ K $ in such a way that we obtain a planar embedding of $ K $. In the embedding, all Steiner points on the line segment $L_k$ will go to the exterior of the polygon along with its adjacent edges. In this planar embedding, $ V' $ still remains in the exterior face. Hence, we get a planar embedding $ K'$ of $ K $ such that $ V' $ is contained in the exterior face of $ K' $.  For an illustration see \cref{fig-k1} and  \cref{fig-k2}.		
	\end{proof}

	\begin{lemma}\label{theo_planarH}
		$ H $ has a planar embedding $ H'$ such that $ V' $ is contained in the exterior face of $ H'$.
	\end{lemma}
	
	\begin{proof} We prove this by weak induction. As $ L_{k-1} $ is horizontal, $ (k-1) $ must be even. Let $ (k-1)=2m $ for some $ m \in \mathbb{N} $. Let $ V_i $ consists of all the vertices in $ G $ on the line segment $ L_{2i} $ and $ G_i $ be the subgraph induced by the vertices lying on or above the line segment $ L_{2i}, $ where $ 2i \leq (k-1) $. So $ G_m= H $. By induction, we prove that $ G_m $ is planar and it has a planar embedding $ H'$ such that $ V' $ is contained in the exterior face of $ H'$. Let $ P(i)$ be the following statement:
		$G_i $ is planar and it has a planar embedding $ G'_{i} $ such that $ V_i $ is contained in the exterior face of  $ G'_{i} $.			
		Now we need to show $ P(m) $ is true. We first show that the base case is true. Next we show the inductive step.
		\begin{figure}[ht!]
			\centering
			\subfigure[]
			{
				\includegraphics[scale=.35]{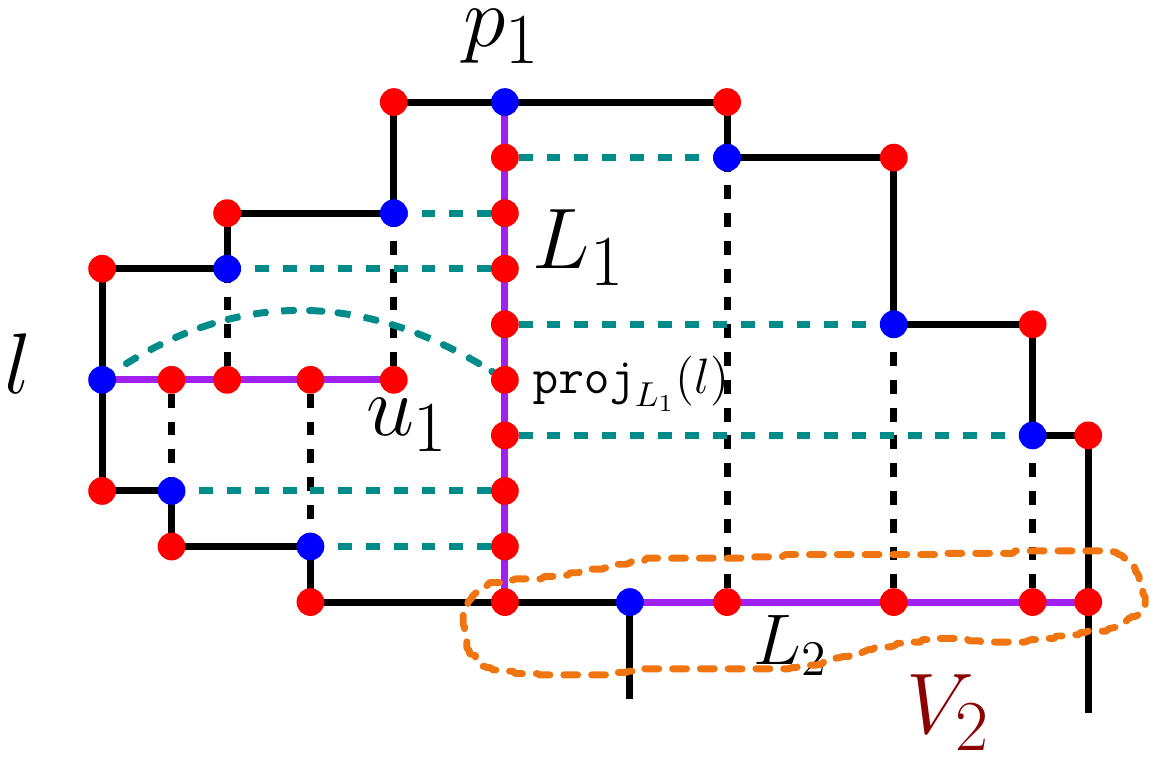}
				\label{fig_g1}
			}
			\hspace{2mm}
			\subfigure[]
			{
				\includegraphics[scale=.35]{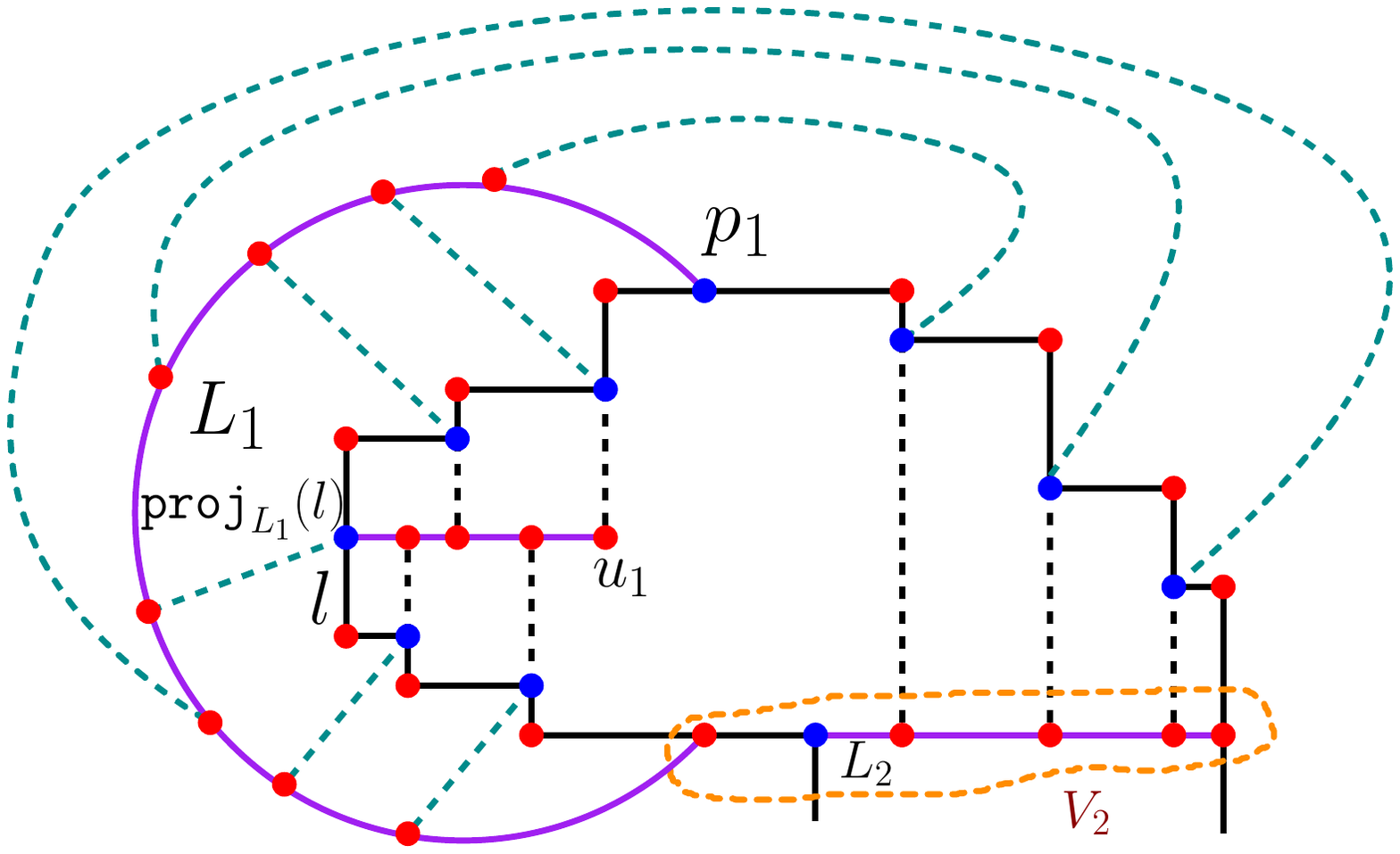}
				\label{fig_g11}
			}
			
			\caption[Optional caption for list of figures]{(a) The graph $ G_1 $. (b) A planar embedding of $ G_1$  with the exterior face containing $ V_2 $. }
			\label{fig_planeg1}
		\end{figure}

		\noindent \underline{\textbf{Base Case:}} $ P(1) $ is true:	  	 We divide the edges of $ G_1$ into three sets $ E_{11}, E_{12},$ and $ E_{13} $. $ E_{11} $ is the set of edges in $ G_1 $ that are along the boundary of the exterior face of $ G_1 $. $ E_{12} $ consists of all the edges in $ G_1 $ that have one endpoint on the segment $ L_1 $.  $ E_{13}= E(G_1) \setminus (E_{11} \cup E_{12}) $. Let $ G_{11} $ be the subgraph of $ G_1$ consisting of the edges $ E_{11} \cup E_{12}, $ and $ G_{12} $ be the subgraph of $ G_1 $ consisting of the edges $ E_{11} \cup E_{13} $. So $ G_1= G_{11} \cup  G_{12} $, where both $ G_{11}$ and $G_{12}$ are plane graphs. In $G_{11}  $ there exists an interior face containing $ V_1 $. Let $V^f_{12}$ be the set of vertices in the exterior face of $ G_{12} $. By \cref{theo_planar}, we can transform the planar embedding $ G_{12} $ into another planar embedding $G'_{12} $ such that there exists an interior face,  say $ f_1 $, that contains $V^f_{12}$. As  $V^f_{12}$ is the set of vertices in the exterior face of $ G_{11}, $ so we can attach $ G_{11} $ in $ f_{1} $ and obtain a planar embedding $ G''_1 $ of $ G_1 $. In $ G''_1 $ there exists an interior face containing $ V_1 $. Applying \cref{theo_planar}, we get a planar embedding $ G'_1$ of $ G_1 $ such that $ V_1 $ is contained in the exterior face of $ G'_1$.	 We illustrate this step in \cref{fig_planeg1}.
		
		\begin{figure}[ht!]
			\centering
			\subfigure[]
			{
				\includegraphics[scale=.3]{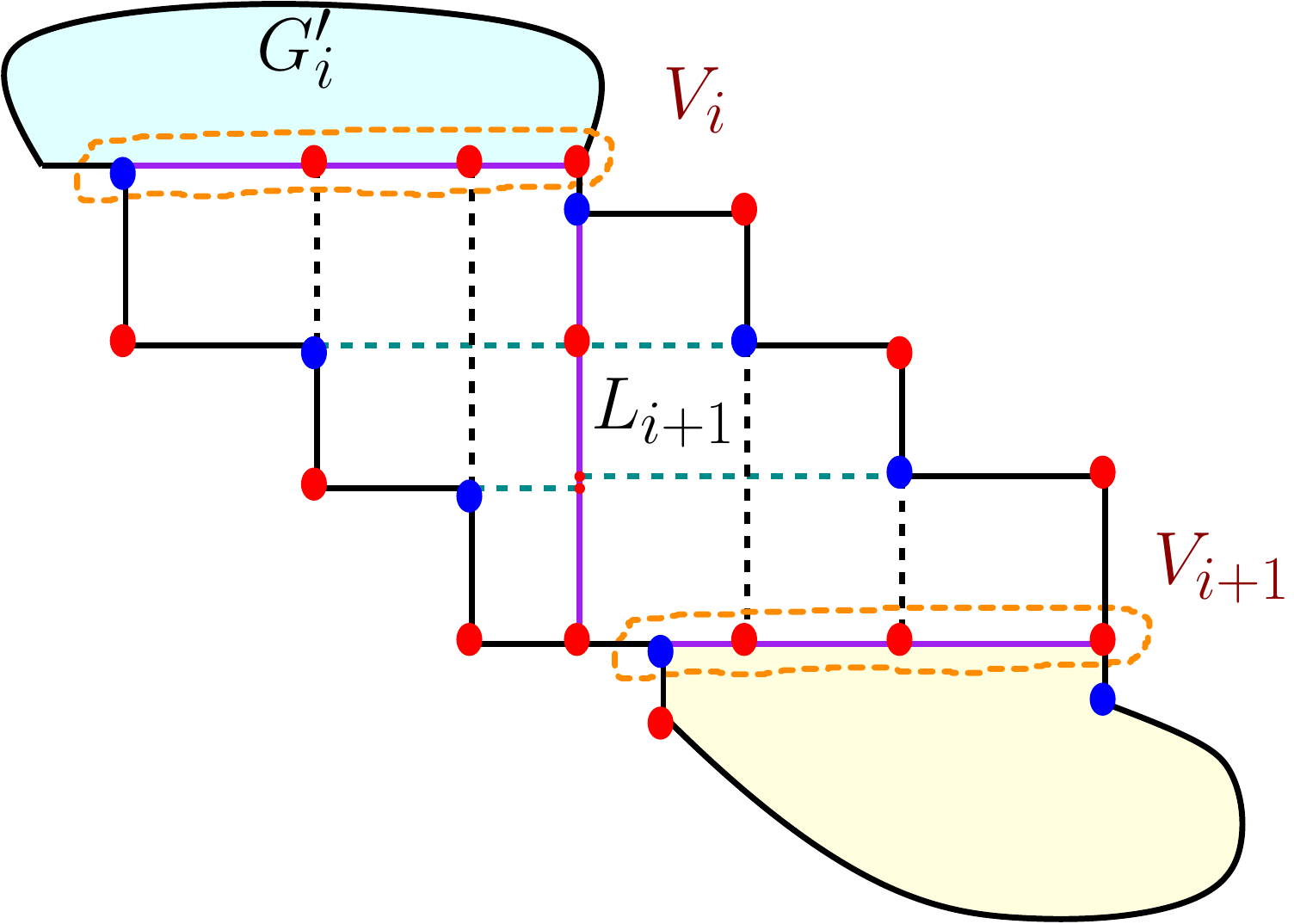}
				\label{fig-i1}
			}
			\hspace{2mm}
			\subfigure[ ]
			{
				\includegraphics[scale=.3]{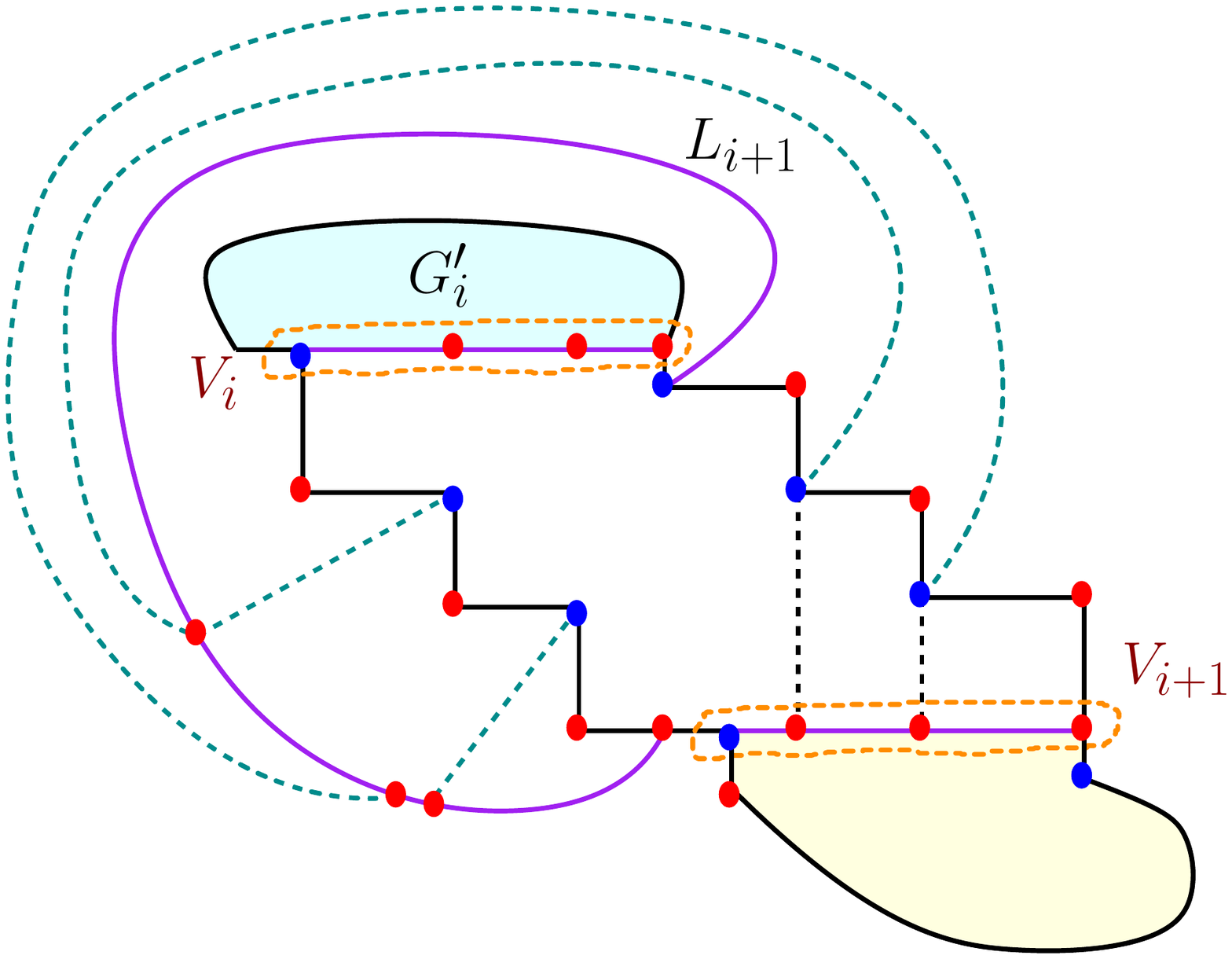}
				\label{fig-i2}
			}
			
			\caption[Optional caption for list of figures]{(a) The graph $ G $ with  planar embedding $G'_i$ having exterior face containing $ V_i $. (b) The graph $ G $ with  planar embedding $G'_{i+1}$ having exterior face containing $ V_{i+1} $.}
			\label{fig-in}
		\end{figure}
	
		\noindent \underline{\textbf{Inductive Case:}} $ P(i) $ is true $ \Rightarrow $ $ P(i+1) $ is true:	   Assume that $ P(i) $ is true, i.e., $ G_i $ has a planar embedding $ G'_{i} $ such that $ V_i $ is contained in the exterior face of  $ G'_{i} $ (see \cref{fig-in}).

		Let $ H_1 $ be the subgraph of $ G_{i+1} $ induced by the vertices lying on or below the line containing $ L_{2i}$. Now $ V_i= G'_i \cap H_1 $, also $ V_i $ is contained in the exterior face of  $ G'_{i} $. As $ G_{i+1}=G_i \cup H_1 $ so in $ G_{i+1} $, we can replace $ G_i $ by its planar embedding $ G'_i $. Now $ G_{i+1}=G'_i \cup H_1 $. Now we divide the edges of $ G_{i+1}$ into three sets $ E_{(i+1)1},E_{(i+1)2},E_{(i+1)3} $. $E_{(i+1)1}$ consists of edges in $ G_{i+1} $ that are along the boundary of the exterior face of $ H_1 $. $E_{(i+1)2}$ consists of edges in $ H_1 $ that have one endpoint on the line containing $ L_{2i+1} $.  $ E_{(i+1)3}= E(H_{1}) \setminus (E_{(i+1)1} \cup E_{(i+1)2}) $. Let $G_{(i+1)1}$ be the subgraph of $ G_{i+1} $ consisting of the edges $ E_{(i+1)1} \cup E_{(i+1)2} \cup E(G'_{i}), $ and $ G_{(i+1)2} $ be the subgraph of $ G_{i+1} $ consisting of the edges $ E_{(i+1)1} \cup E_{(i+1)3} \cup E(G'_{i})$. So $ G_{i+1}= G_{(i+1)1} \cup  G_{(i+1)2},$ where both $ G_{(i+1)1}$ and $G_{(i+1)2}$ are plane graphs. In $  G_{(i+1)1}$ there exists an interior face containing $ V_{i+1} $. Let $V^f_{(i+1)2}$ be the set of vertices in the exterior face of $  G_{(i+1)2} $. By \cref{theo_planar}, we can transform the planar embedding $  G_{(i+1)2} $ into another planar embedding $  G'_{(i+1)2} $ such that there exists an interior face,  say $ f $, that contains $V^f_{(i+1)2}$. As  $V^f_{(i+1)2}$ is also the set of vertices in the exterior face of $  G_{(i+1)1} $, so we can    attach $  G_{(i+1)1} $ in $ f $ and obtain a planar embedding $  G''_{i+1} $ of $  G_{i+1} $. In $  G''_{i+1} $  there exists an interior face containing $ V_{i+1} $. Applying \cref{theo_planar}, we get our desired planar embedding $ G'_{i+1}$ of $ G_{i+1} $ such that $ V_{i+1}$ is contained in the exterior face of $  G'_{i+1} $.

		Now by the induction hypothesis, $ P(m) $ is true, i.e., $G_m $ is planar and it has a planar embedding $ G'_{m} $ such that $ V_m $ is contained in the exterior face of  $ G'_{m} $. Now $ V_m $ consists of all the vertices on the line $ L_{2m} $. Now $ 2m=k $ implies that $ V_m=V' $. Also $ G_m=H $. So $ H $ has a planar embedding $ H'(=G'_m)$ such that $ V' $ is contained in the exterior face of $ H'$. 
	\end{proof}
	\noindent \underline{\textbf{Case 2. $ \bm{L_{k-1} }$ is vertical$\bm{\colon}$}}  
	Proof of the planarity of $G$ for this case is similar to Case 1.	When $ L_{k-1} $ is vertical, we partition $ G $ into $ H $ and $ K $ as follows:  $ H $ and $ K $ are the subgraphs of $ G $ induced by the vertices lying to the left and right, respectively of the line $ L_{k-1} $.    Both $ H $ and $ K $ must include the vertices on $ L_{k-1} $. Here, we only prove planarity for $ K $. The rest of proof  is similar to case 1.
	
	\begin{figure}[ht!]
		\hspace{5mm}
		\subfigure[]
		{
			\includegraphics[scale=0.35]{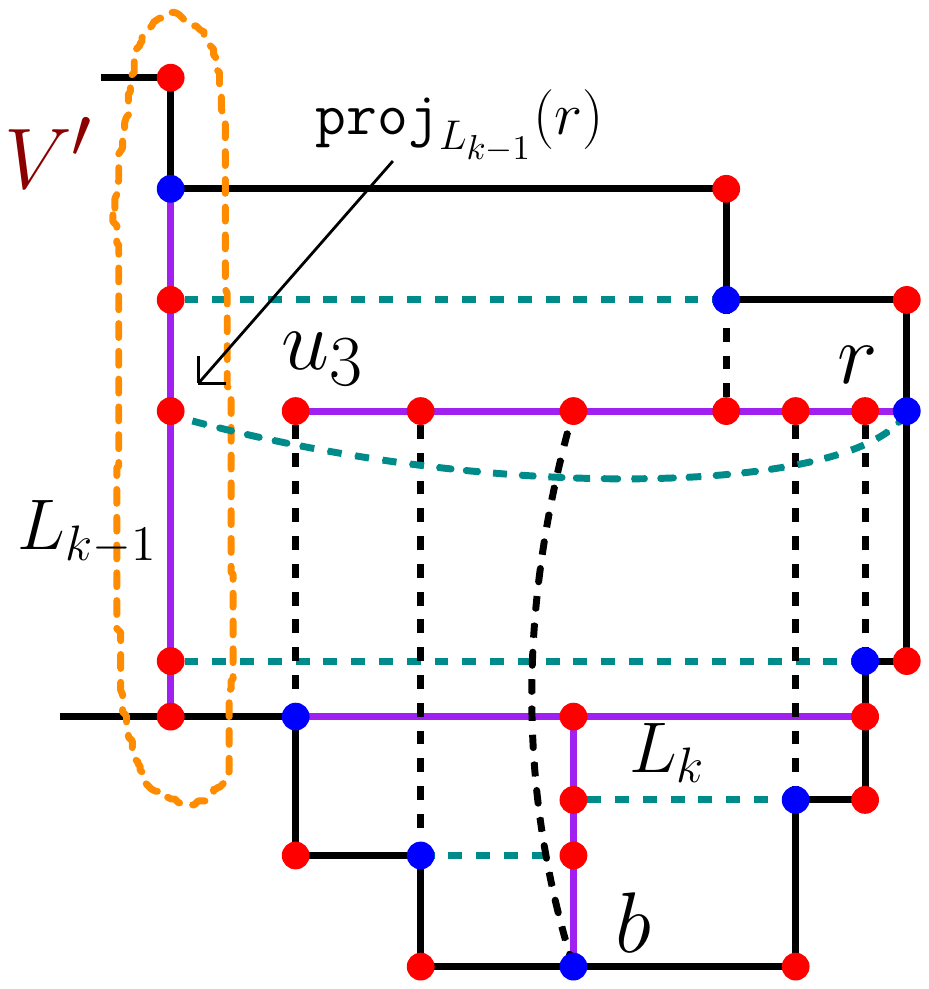}
			\label{}
		}
		\hspace{15mm}	
		\subfigure[]
		{
			\includegraphics[scale=0.35]{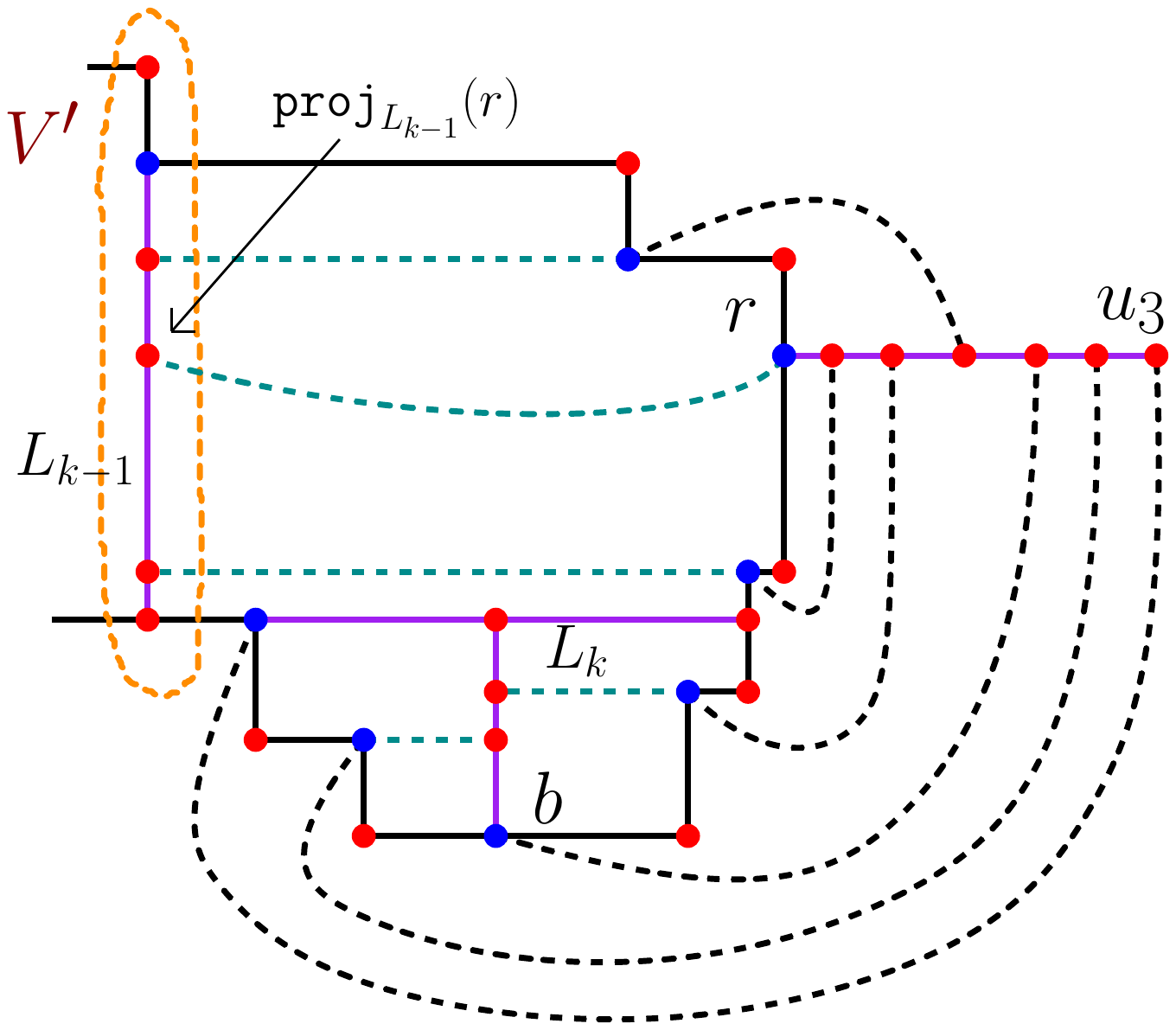}
			\label{}
		}
		
		\caption{ (a) The subgraph $ \bm{K} $ of $ G $ for \textbf{Type 3} with the exterior face containing $ V' $. (b) Planar embedding of $ \bm{K} $. The edges of $E_r$ are shown by dashed black segment.}
		\label{fig-k3}
	\end{figure}
	
	\begin{figure}[ht!]
		\hspace{5mm}
		\subfigure[]
		{
			\includegraphics[scale=0.35]{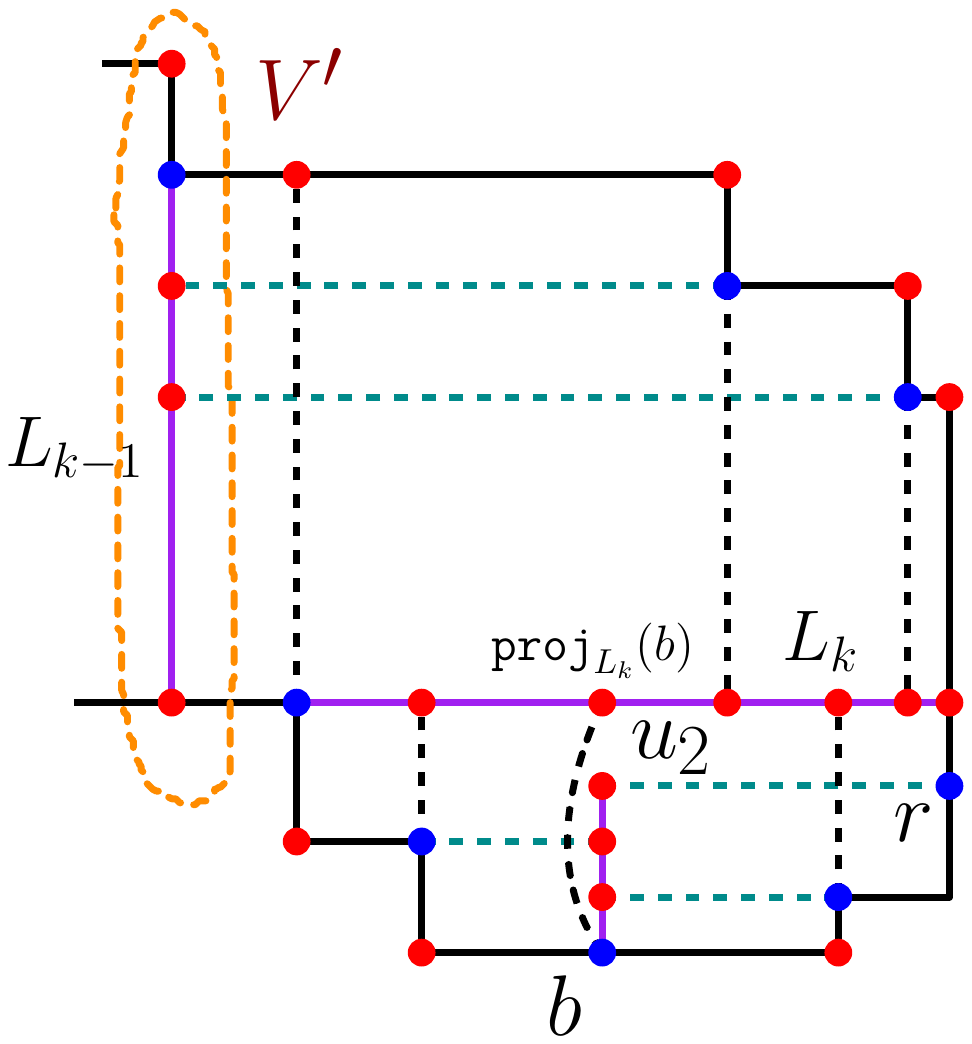}
			\label{}
		}
		\hspace{15mm}	
		\subfigure[]
		{
			\includegraphics[scale=0.35]{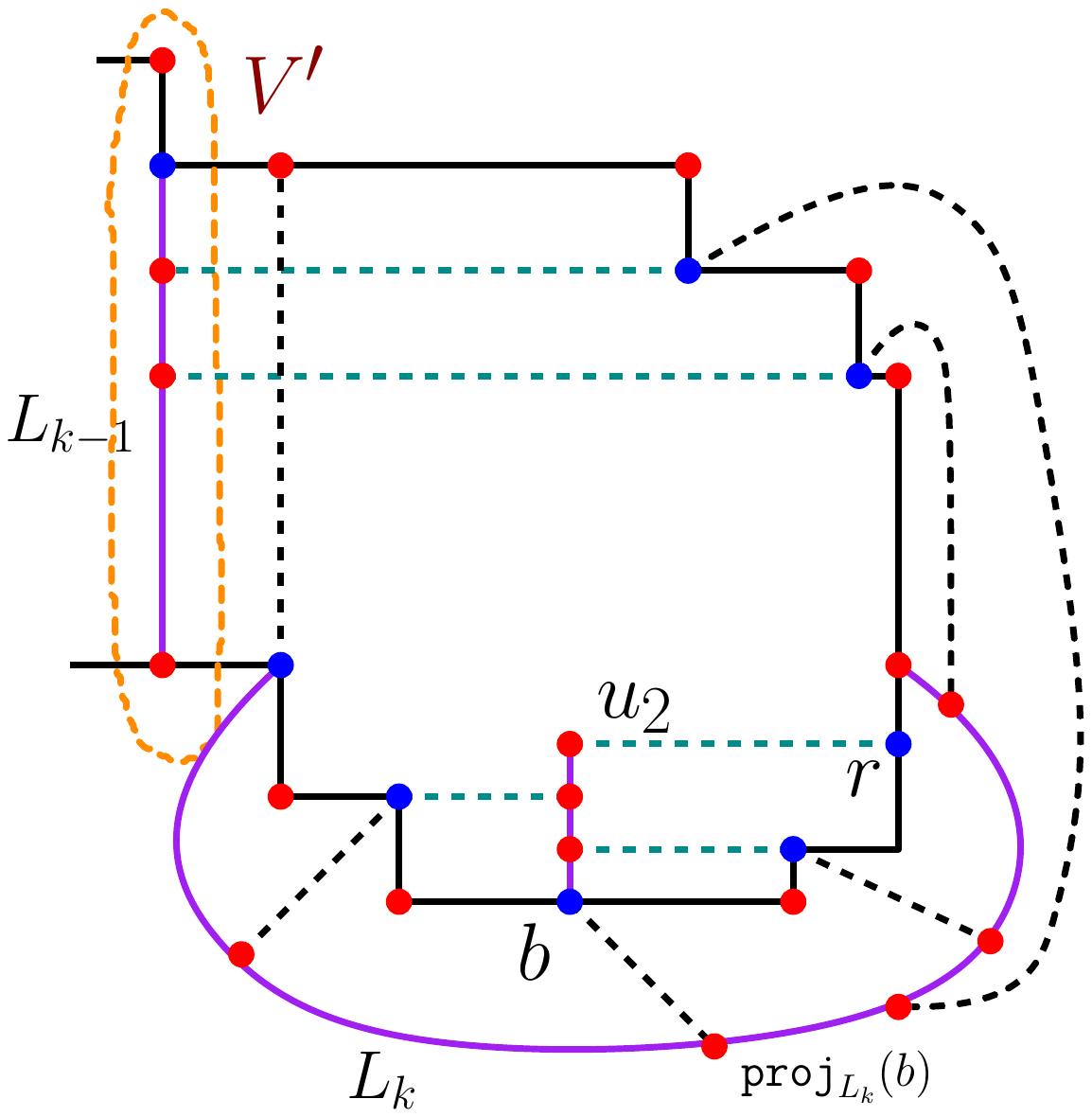}
			\label{}
		}
		
		\caption{(a) The subgraph $ \bm{K} $ of $ G $ for\textbf{ Type 4} with the exterior face containing $ V' $. (b) Planar embedding of  $ \bm{K} $. The edges of $E_k$ are shown by dashed black segment.}
		\label{fig-k4}
	\end{figure}
	
	\begin{lemma}
		$ K $ has a planar embedding $ K'$ such that $ V' $ is contained in the exterior face of   $ K'$.
	\end{lemma}
	
	\noindent{\emph{Proof.}} Let $ V^f_k \subseteq V$ be the set of vertices in $ G $ along the exterior face of $ K $. So $ V^f_k = (V^f \cap V(K)) \cup V'$. For Type 3, let $ E_r $ be the set of vertical edges that have at least one adjacent vertex on the line $ e'_3= \overline{r~\mathtt{proj}_{L_{k-1}}(r)} $. In this case, we draw the edges $ E_r $ in the exterior faces of $ K $ in such a way that we obtain a planar embedding of $ K $. In the embedding, all Steiner points on the line segment $\overline{ru_3}$ will go to the exterior of the polygon along with its adjacent edges (see \cref{fig-k3}). For Type 4, let $ E_k $ be the set of vertical edges that have at least one adjacent vertex on the line $ L_k $. In this case, we draw the edges $ E_k $ in the exterior faces of $ K $ in such a way that we obtain a planar embedding of $ K $. In the embedding, all Steiner points on the line segment $L_k$ will go to the exterior of the polygon along with its adjacent edges (see \cref{fig-k4}). In this planar embedding $ V' $ still remains in the exterior face. Hence, we get a planar embedding $ K'$ of $ K $ such that $ V' $ is contained in the exterior face of $ K' $.

\end{proof}

\section{Conclusion}\label{sec-con}
In this paper, we construct a \planar~$G$ for a  given convex point set $S$ of size $ n $ in linear time, where $ G $ contains  $ \mathcal{ O}(n) $  Steiner points. Our construction works for more general point set where it is possible to construct an ortho-convex polygon $ \mathcal{OCP}(S) $ such that $S$ lies on the boundary of $ \mathcal{OCP}(S) $. For example, any convex point set satisfies the aforesaid property. It is also clear that there exists convex point set $S$ for which \planar~$G$ needs $\Omega(n)$ Steiner points. Let $S=\{(1,1),(2,2),\ldots,(n,n)\}$ be a convex point set of size $ n $. Then $S$ would need $\Omega(n)$ Steiner points. In that sense, our construction is optimal for convex point sets. As a corollary of our construction, for a convex point set, we obtain a $\sqrt2 (\sim 1.41)$ planar spanner in $L_2$ norm using $\mathcal{O}(n)$ Steiner points. It remains an open question, if it is possible to construct a \planar~for general point sets using subquadratic number of Steiner points.

\section*{References}

\bibliography{cgta}

\begin{thebibliography}{10}
\expandafter\ifx\csname url\endcsname\relax
  \def\url#1{\texttt{#1}}\fi
\expandafter\ifx\csname urlprefix\endcsname\relax\def\urlprefix{URL }\fi
\expandafter\ifx\csname href\endcsname\relax
  \def\href#1#2{#2} \def\path#1{#1}\fi

\bibitem{gudmundsson2007small}
J.~Gudmundsson, O.~Klein, C.~Knauer, M.~Smid, Small \uppercase{m}anhattan
  networks and algorithmic applications for the \uppercase{e}arth mover’s
  distance, in: Proceedings of the 23rd European Workshop on Computational
  Geometry, 2007, pp. 174--177.

\bibitem{gudmundsson1999approximating}
J.~Gudmundsson, C.~Levcopoulos, G.~Narasimhan, Approximating minimum
  \uppercase{m}anhattan networks, in: Randomization, Approximation, and
  Combinatorial Optimization. Algorithms and Techniques, Springer, 1999, pp.
  28--38.

\bibitem{narasimhan2007geometric}
G.~Narasimhan, M.~Smid, Geometric spanner networks, Cambridge University Press,
  2007.

\bibitem{lam2003picking}
F.~Lam, M.~Alexandersson, L.~Pachter, Picking alignments from (steiner) trees,
  Journal of Computational Biology 10~(3-4) (2003) 509--520.

\bibitem{kato2002improved}
R.~Kato, K.~Imai, T.~Asano, An improved algorithm for the minimum
  \uppercase{m}anhattan network problem, in: International Symposium on
  Algorithms and Computation, Vol. 2518 of LNCS, Springer, 2002, pp. 344--356.

\bibitem{seibert20051}
S.~Seibert, W.~Unger, A 1.5-approximation of the minimal \uppercase{m}anhattan
  network problem, in: International Symposium on Algorithms and Computation,
  Vol. 3827 of LNCS, Springer, 2005, pp. 246--255.

\bibitem{benkert2004minimum}
M.~Benkert, A.~Wolff, F.~Widmann, The minimum \uppercase{m}anhattan network
  problem: a fast factor-3 approximation, in: Japanese Conference on Discrete
  and Computational Geometry, Springer, 2004, pp. 16--28.

\bibitem{guo2008greedy}
Z.~Guo, H.~Sun, H.~Zhu, Greedy construction of 2-approximation minimum
  \uppercase{m}anhattan network, in: International Symposium on Algorithms and
  Computation, Vol. 5369 of LNCS, Springer, 2008, pp. 4--15.

\bibitem{chin2011minimum}
F.~Y. Chin, Z.~Guo, H.~Sun, Minimum \uppercase{m}anhattan network is
  \uppercase{NP}-complete, Discrete \& Computational Geometry 45~(4) (2011)
  701--722.

\bibitem{knauer2011fixed}
C.~Knauer, A.~Spillner, A fixed-parameter algorithm for the minimum
  \uppercase{m}anhattan network problem., Journal of Computational Geometry
  2~(1).

\bibitem{Keil1989}
J.~M. Keil, C.~A. Gutwin, The delaunay triangulation closely approximates the
  complete euclidean graph, in: Lecture Notes in Computer Science, Springer
  Berlin Heidelberg, 1989, pp. 47--56.

\bibitem{Cui2011}
S.~Cui, I.~A. Kanj, G.~Xia, On the stretch factor of delaunay triangulations of
  points in convex position, Computational Geometry 44~(2) (2011) 104--109.

\bibitem{Xia2013}
G.~Xia, The stretch factor of the delaunay triangulation is less than 1.998,
  {SIAM} Journal on Computing 42~(4) (2013) 1620--1659.

\bibitem{arikati1996planar}
S.~Arikati, D.~Z. Chen, L.~P. Chew, G.~Das, M.~Smid, C.~D. Zaroliagis, Planar
  spanners and approximate shortest path queries among obstacles in the plane,
  in: European Symposium on Algorithms, Vol. 1136 of LNCS, Springer, 1996, pp.
  514--528.

\bibitem{amani2016plane}
M.~Amani, A.~Biniaz, P.~Bose, J.-L. De~Carufel, A.~Maheshwari, M.~Smid, A plane
  1.88-spanner for points in convex position, Journal of Computational Geometry
  7~(1) (2016) 520--539.

\bibitem{schuierer1996optimal}
S.~Schuierer, An optimal data structure for shortest rectilinear path queries
  in a simple rectilinear polygon, International Journal of Computational
  Geometry \& Applications 6~(02) (1996) 205--225.

\bibitem{de1991rectilinear}
M.~De~Berg, On rectilinear link distance, Computational Geometry 1~(1) (1991)
  13--34.

\bibitem{datta1990some}
A.~Datta, G.~Ramkumar, On some largest empty orthoconvex polygons in a point
  set, in: International Conference on Foundations of Software Technology and
  Theoretical Computer Science, Vol. 472 of LNCS, Springer, 1990, pp. 270--285.

\bibitem{chepoi2008rounding}
V.~Chepoi, K.~Nouioua, Y.~Vaxes, A rounding algorithm for approximating minimum
  \uppercase{m}anhattan networks, Theoretical Computer Science 390~(1) (2008)
  56--69.

\bibitem{west1996introduction}
D.~B. West, et~al., Introduction to graph theory, Vol.~2, Prentice Hall Upper
  Saddle River, NJ, 1996.

\bibitem{gibbons1985algorithmic}
A.~Gibbons, Algorithmic graph theory, Cambridge university press, 1985.

\bibitem{garcia2015geometric}
A.~Garc{\'\i}a, F.~Hurtado, M.~Korman, I.~Matos, M.~Saumell, R.~I. Silveira,
  J.~Tejel, C.~D. T{\'o}th, Geometric biplane graphs ii: Graph augmentation,
  Graphs and Combinatorics 31~(2) (2015) 427--452.

\end{thebibliography}

\end{document}